\documentclass[twocolumn,pra,floatfix,amsmath,superscriptaddress]{revtex4}
\usepackage[latin9]{inputenc}
\usepackage{amsmath}
\usepackage{amsthm}
\usepackage{amssymb}
\usepackage[unicode=true,pdfusetitle,
 bookmarks=true,bookmarksnumbered=false,bookmarksopen=false,
 breaklinks=false,pdfborder={0 0 0},backref=false,colorlinks=false]
 {hyperref}
\hypersetup{
 colorlinks,linkcolor=myurlcolor,citecolor=myurlcolor,urlcolor=myurlcolor}

\makeatletter
\@ifundefined{textcolor}{}
{%
 \definecolor{BLACK}{gray}{0}
 \definecolor{WHITE}{gray}{1}
 \definecolor{RED}{rgb}{1,0,0}
 \definecolor{GREEN}{rgb}{0,1,0}
 \definecolor{BLUE}{rgb}{0,0,1}
 \definecolor{CYAN}{cmyk}{1,0,0,0}
 \definecolor{MAGENTA}{cmyk}{0,1,0,0}
 \definecolor{YELLOW}{cmyk}{0,0,1,0}
}
\theoremstyle{plain}
\newtheorem{thm}{\protect\theoremname}
  \theoremstyle{definition}
  \newtheorem{defn}[thm]{\protect\definitionname}
  \theoremstyle{plain}
  \newtheorem{cor}[thm]{\protect\corollaryname}
  \theoremstyle{plain}
  \newtheorem{lem}[thm]{\protect\lemmaname}

\usepackage{color}

\newcommand{\bea}{\begin{eqnarray}}
\newcommand{\eea}{\end{eqnarray}}

\def\bi{\begin{itemize}}
\def\ei{\end{itemize}}
\def\bc{\begin{center}}
\def\ec{\end{center}}

\def\C{\hbox{$\mit I$\kern-.7em$\mit C$}}
\def\R{\hbox{$\mit I$\kern-.6em$\mit R$}}

\def\ket#1{|#1\rangle}

\def\tr{\mathrm{tr}}
\def\ket#1{\left| #1\right>}
\def\bra#1{\left< #1\right|}
\def\bk#1{\langle #1 \rangle}
\newcommand{\proj}[1]{\ket{#1}\bra{#1}}

\definecolor{myurlcolor}{rgb}{0,0,0.7}
\usepackage{pxfonts}

\providecommand{\definitionname}{Definition}
\providecommand{\theoremname}{Theorem}

\providecommand{\definitionname}{Definition}
\providecommand{\theoremname}{Theorem}

\providecommand{\definitionname}{Definition}
\providecommand{\theoremname}{Theorem}

\providecommand{\definitionname}{Definition}
\providecommand{\theoremname}{Theorem}

\providecommand{\corollaryname}{Corollary}
\providecommand{\definitionname}{Definition}
\providecommand{\lemmaname}{Lemma}
\providecommand{\theoremname}{Theorem}

\providecommand{\corollaryname}{Corollary}
\providecommand{\definitionname}{Definition}
\providecommand{\lemmaname}{Lemma}
\providecommand{\theoremname}{Theorem}

\providecommand{\corollaryname}{Corollary}
\providecommand{\definitionname}{Definition}
\providecommand{\lemmaname}{Lemma}
\providecommand{\theoremname}{Theorem}

\providecommand{\corollaryname}{Corollary}
\providecommand{\definitionname}{Definition}
\providecommand{\lemmaname}{Lemma}
\providecommand{\theoremname}{Theorem}

\providecommand{\corollaryname}{Corollary}
\providecommand{\definitionname}{Definition}
\providecommand{\lemmaname}{Lemma}
\providecommand{\theoremname}{Theorem}

\makeatother

  \providecommand{\corollaryname}{Corollary}
  \providecommand{\definitionname}{Definition}
  \providecommand{\lemmaname}{Lemma}
\providecommand{\theoremname}{Theorem}

\begin{document}

\title{The power of the resource theory of genuine quantum coherence}

\author{Julio I. de Vicente}

\affiliation{Departamento de Matemáticas, Universidad Carlos III de Madrid, Avda.\
de la Universidad 30, E-28911, Leganés (Madrid), Spain}

\author{Alexander Streltsov}

\affiliation{Dahlem Center for Complex Quantum Systems,\\
 Freie Universität Berlin, D-14195 Berlin, Germany}

\date{April 18, 2016}
\begin{abstract}
Any quantum resource theory is based on free states and free operations,
i.e., states and operations which can be created and performed at
no cost. In the resource theory of coherence free states are diagonal
in some fixed basis, and free operations are those which cannot create
coherence for some particular experimental realization. Recently,
some problems of this approach have been discussed, and new sets of
operations have been proposed to resolve these problems. As an example
for such a problem, a quantum operation can be incoherent in one particular
experimental realization, but it still can create large amount of
coherence in another. This problem is resolved by the framework of
genuine quantum coherence: the free operations proposed there preserve
all diagonal states, and it was shown that such operations cannot
create coherence regardless of their experimental realization. Here,
we investigate this concept further and study other more general sets
of operations which cannot create coherence in any experimental realization.
We analyze in detail the mathematical structure of this class of maps
and use this to study possible state transformations under these operations.
We show that deterministic manipulation is severely limited, even
in the asymptotic settings. In particular, this framework does not
have a unique golden unit, i.e., there is no single state from which
all other states can be created deterministically with the free operations.
On the other hand, we show that these theories allow for a richer
structure in what comes to stochastic manipulation of states. In particular,
we determine the optimal probability of conversion among pure states
with genuinely incoherent operations. Still, the aforementioned limitations
suggest that any reasonably powerful resource theory of coherence
must contain free operations which can potentially create coherence
in some experimental realization.
\end{abstract}
\maketitle

\section{Introduction}

Quantum mechanics offers a radically different description of reality
that collides with the intuition behind that of classical physics.
At first this was only regarded from the foundational point of view.
However, in recent decades it has been realized that the fundamentally
different features of quantum theory can be exploited to realize revolutionary
applications \cite{Nielsen10}. Quantum information theory has taught
us that quantum technologies can outperform classical ones in a large
variety of tasks such as communication, computation or metrology.
This has led to identify and study nonclassical salient properties
of quantum theory, like entanglement \cite{Plenio2007,HorodeckiRMP09}
or nonlocality \cite{Brunner2014} which stem from the tensor product
structure. This is done in order to understand better from a theoretical
perspective the full potential of quantum resources and, also, to
seek for new paths for applications. Undoubtedly, the superposition
principle, which leads to coherence, is another characteristic trait
of quantum mechanics; however, a rigorous theoretical study of this
phenomenon on the analogy of the aforementioned resources has only
been initiated very recently \cite{Aberg2006,Gour2008,Gour2009,Levi14,Marvian2013,Marvian2014,Baumgratz2014,Winter2015}.
Nevertheless, quantum coherence is the basis of single-particle interferometry
\cite{Oi2003,Aberg2004,Oi2006} and it is believed to play a nontrivial
role in the outstanding efficiency of several biological processes
\cite{Lloyd2011,Li2012,Huelga13,Singh2014}. This grants coherence
the status of a resource and makes necessary to develop a solid framework
allowing to asses and quantify this phenomenon together with the rules
for its manipulation.

Resource theories have proven to be a very successful framework to
build a rigorous and systematic study of the possibilities and limitations
of distinct features of quantum information theory. Originally developed
in the case of entanglement theory \cite{Plenio2007,HorodeckiRMP09},
the conceptual elegance and applicability of this approach has led
to consider in the last years, resource theories for several quantum
features such as frame alignment \cite{Gour2008}, stabilizer computation
\cite{Veitch2014}, nonlocality \cite{DeVicente2014} or steering
\cite{Gallego2015}. In such theories one considers a set of free
states and of free operations. The latter constitutes the set of transformations
that the physical setting allows to implement. Free states must be
mapped to free states under all free operations and they are useless
in this physical setting (it is usually assumed that they can be prepared
at no cost). With this, non-free states can be regarded as resource
states: they allow to overcome the limitations imposed by state manipulation
under the set of free operations. Furthermore, free operations then
provide all possible protocols to manipulate the resource and induce
the most natural ordering among states since the resource cannot increase
under this set of transformations. This allows to rigorously construct
resource measures: these quantifiers must not increase under free
operations. For instance, in entanglement theory the set of free operations
is local (quantum) operations and classical communication (LOCC) while
free states are separable states (entangled states are then the resource
states) and the basic principle behind entanglement measures is that
they must not increase under LOCC \cite{Plenio2007,HorodeckiRMP09}.
Recent literature has set the first steps to build a resource theory
of coherence \cite{Aberg2006,Gour2008,Gour2009,Levi14,Marvian2013,Marvian2014,Baumgratz2014,Winter2015},
which has allowed to study the role of coherence in quantum theory
\cite{Chen2015,Cheng2015,Hu2015,Mondal2015,Deb2015,Du2015b,Bai2015,Bera2015,Liu2015,Xi2014,Zou2014,Prillwitz2014,Karpat2014,Cakmak2015,Du2015a,Hillery2015,Yadin2015a,Bu2015b,Matera2015,Chitambar2015,Chitambar2015a,Ma2015,Streltsov2016},
its dynamics under noisy evolution \cite{Bromley2015,Singh2015b,Mani2015,Singh2015,Chanda2015,Bu2015,Singh2015a,Zhang2015,Allegra2015,Garcia-Diaz2015,Silva2015,Peng2015},
and to obtain new coherence measures \cite{Girolami2014,Du2015,Streltsov2015b,Killoran2015,Girolami2015,Yuan2015,Yao2015,Pires2015,Qi2015,Xu2015,Yadin2015,Rana2016,Zhang2015b,Chandrashekar2016,Chen2016,Chitambar2016,Streltsov2015d,Napoli2016,Piani2016}.

However, there is an ongoing intense debate on how this theory should
be exactly formulated with several alternatives being considered \cite{Aberg2006,Gour2009,Marvian2014a,Marvian2015,Winter2015,Yadin2015a,Chitambar2016,Marvian2016}.
Notice that in the case of entanglement the physical setting clearly
identifies the set of allowed operations: the parties, who might be
spatially separated, can only act locally. However, while in the case
of coherence it is clear that free states should correspond to incoherent
states (see below for definitions), the physical setting does not
impose any clear restriction on what free operations should be. Thus,
any set of operations that map incoherent states to incoherent states
might qualify in principle as a good candidate. It is therefore fundamental
to identify a set of free operations based on a clear physical interpretation.
In this paper, we analyze in detail two such sets and we thoroughly
study the possibilities and limitations of the emergent resource theories
for state manipulation. This allows us to clarify what the ultimate
resource theory of coherence should be.

In the framework of coherence, the physical setting identifies a particular
set of basis states as classical. A state on $H\simeq\mathbb{C}^{d}$
is called incoherent if it is diagonal in the fixed aforementioned
basis $\{|i\rangle\}$ ($i=1,\ldots,d$) and otherwise coherent. In
the standard resource theory of coherence of Baumgratz \emph{et al.}
\cite{Baumgratz2014}, free operations are given by the so-called
incoherent operations. These correspond to those maps that admit an
incoherent Kraus decomposition:
\begin{equation}
\Lambda_{\mathrm{i}}[\rho]=\sum_{i}K_{i}\rho K_{i}^{\dagger},\label{incoherentoperations}
\end{equation}
where, besides the normalization condition $\sum_{i}K_{i}^{\dagger}K_{i}=\openone$,
the Kraus operators fulfill that the unnormalized states $K_{i}\rho K_{i}^{\dagger}$
remain incoherent for every $i$ if $\rho$ is. On an experimental
level, this definition means that a quantum operation \emph{can} be
implemented in an incoherent way. By performing the measurements given
by the aforementioned Kraus operators, no coherence can be created
from an incoherent state even if we allow for postselection. However,
in general there are many inequivalent experimental realizations leading
to the same quantum operation. This also means that a quantum operation
which is incoherent in some Kraus decomposition can still create coherence
in another \cite{Streltsov2015d}. Even worse, there are examples
of operations which, while being incoherent in one experimental realization,
create maximal coherence in another \cite{Streltsov2015d}. A similar
observation can be made for LOCC maps in entanglement theory. Notwithstanding,
in this case it is clear that this is irrelevant since the physical
setting only cares for the existence of a local implementation. However,
in the case of coherence it is not clear whether allowed transformations
should be defined at the level of maps or of particular implementations.
It seems therefore necessary to study the full potential of a resource
theory of coherence built on the constraint that the allowed set of
operations are incoherent maps.

The framework of genuine coherence was introduced with the aim to
rule out this sort of ``hidden coherence'' \cite{Streltsov2015d}.
The main ingredient of this framework is the definition of genuinely
incoherent (GI) operations. These operations preserve all incoherent
states:
\begin{equation}
\Lambda_{\mathrm{gi}}[\rho_{\mathrm{i}}]=\rho_{\mathrm{i}}
\end{equation}
for any incoherent state $\rho_{\mathrm{i}}$. GI operations are particularly
meaningful from the physical point of view in scenarios where the
classical states are not interchangeable, e.\ g.\ when they have
different energies. This also implies that the framework of genuine
coherence falls into the same class as the resource theory of asymmetry
\cite{Gour2009,Marvian2014a,Marvian2015} and characterizes \emph{unspeakable}
coherence \cite{Marvian2016}, while the framework of Baumgratz \emph{et
al.} \cite{Baumgratz2014} is a representative of \emph{speakable}
coherence. As was shown in \cite{Streltsov2015d}, the Kraus operators
of GI operations are all diagonal in the incoherent basis. Moreover,
the latter property does not depend on a particular Kraus decomposition,
and thus any Kraus decomposition of such an operation is indeed incoherent.
This means that a genuinely incoherent quantum operation cannot create
coherence, regardless of its particular experimental realization.
Reference \cite{Streltsov2015d} uses this framework as a starting
point for a resource theory of genuine coherence and provides measures
for such resource. Interestingly, results therein (cf.\ Theorem 6)
seem to indicate that this theory might have limitations for deterministic
state manipulation.

In this paper we analyze in detail the mathematical structure of GI
operations. This allows us to study extensively their possibilities
and limitations for state manipulation. We show that deterministic
transformations among pure states are impossible in this framework.
However, all mixed states can be obtained from some pure state and
we characterize this kind of transformations. Moreover, we study stochastic
transformations as well, which have a much richer structure, and we
obtain the maximal probability of transformation among arbitrary pure
states. We also consider the asymptotic setting, which plays a key
role in resource theories, and show that any general form of distillation
and dilution is impossible. Motivated by these limitations we introduce
and study the set of fully incoherent (FI) operations. This is defined
as all operations in which all Kraus decompositions are incoherent,
i.\ e. they cannot create coherence in any experimental implementation.
We characterize mathematically this kind of operations, which allows
us to prove that this set is a strict superset of GI operations and
to study its potential for state manipulation. We show that in this
case deterministic transformations among pure coherent states are
indeed possible. However, we provide several results proving that
state manipulation under FI operations is still rather contrived outside
the stochastic case. Finally, we provide a discussion on the relation
of the operations presented in this work to other alternative incoherent
operations presented in the literature.

\section{Genuinely incoherent operations}

\subsection{\label{sec:Schur}Relation to Schur operations}

The results presented in this section are strongly based on the framework
of \emph{Schur operations}. A quantum operation $\Lambda$ acting
on a Hilbert space of dimension $d$ is called a Schur operation if
there exists a $d\times d$ matrix $A$ such that
\begin{equation}
\Lambda[\rho]=A\odot\rho.
\end{equation}
Here, $\odot$ denotes the Schur or Hadamard product, i.\ e.\ entry-wise
product for two matrices of the same dimension:
\begin{equation}
(X\odot Y)_{ij}=X_{ij}Y_{ij}.
\end{equation}

The fact that $\Lambda$ is a quantum channel -- and thus a trace
preserving completely positive map -- adds additional constraints
on the matrix $A$ \cite{Watrous2011}:
\begin{itemize}
\item $A$ must be positive semidefinite (PSD),
\item the diagonal elements of $A$ must be $A_{ii}=1$.
\end{itemize}
The following theorem provides an important relation between Schur
operations and GI operations that will be used throughout this section.
\begin{thm}
\label{thm:Schur}The following statements are equivalent:
\begin{enumerate}
\item $\Lambda$ is a genuinely incoherent quantum operation, i.\ e.\ $\Lambda[\rho]=\rho$
for every incoherent state $\rho$.
\item Any Kraus representation of $\Lambda$ as in Eq.\ (\ref{incoherentoperations})
has all Kraus operators $\{K_{i}\}$ diagonal.
\item $\Lambda$ can be written as
\begin{equation}
\Lambda[\rho]=A\odot\rho\label{eq:Schur}
\end{equation}
with a PSD matrix $A$ such that $A_{ii}=1$.
\end{enumerate}
\end{thm}
\begin{proof}
The equivalence of 1 and 2 was shown in \cite{Streltsov2015d}. Moreover,
it is straightforward to verify by inspection that any operation defined
as in Eq.~(\ref{eq:Schur}) is indeed genuinely incoherent, i.e.,
it preserves all incoherent states. This implies that $3\Rightarrow1$.
It remains to prove that for any GI operation $\Lambda$ there exists
a matrix $A$ such that Eq.~(\ref{eq:Schur}) holds true. This is
a direct consequence of Theorem 4.19 in \cite{Watrous2011}.
\end{proof}
This theorem provides a simple characterization of genuinely incoherent
operations, and will be used for proving several statements in this
section. As we will see in the next subsection, the theorem also generalizes
to nondeterministic transformations. Furthermore, this characterization
allows to study in great detail the mathematical structure of the
set of GI maps and to obtain answers to several questions left open
in \cite{Streltsov2015d}. However, these properties are not needed
to establish the different results on state manipulation that we present
in Sec.\ \ref{GImanipulationsec} and that constitute the main aim
of this paper. For this reason we refer the reader interested in other
properties of GI maps to Appendix \ref{sec:MathematicalStructure}.

\subsection{Stochastic genuinely incoherent operations }

We will now consider trace non-increasing genuinely incoherent operations,
i.\ e.\ transformations of the form
\begin{equation}
\Lambda_{\mathrm{sgi}}[\rho]=\sum_{i}K_{i}\rho K_{i}^{\dagger},
\end{equation}
where the Kraus operators $K_{i}$ are all diagonal in the incoherent
basis but do not need to form a complete set, i.e., $\sum_{i}K_{i}^{\dagger}K_{i}\leq\openone$.
This means that the process is not deterministic, but occurs with
probability given by $p=\mathrm{Tr}[\sum_{i}K_{i}\rho K_{i}^{\dagger}]$.
We will call such a map \textbf{\emph{s}}\emph{tochastic }\textbf{\emph{g}}\emph{enuinely
}\textbf{\emph{i}}\emph{ncoherent} \emph{(SGI)} operation. It is important
to note that any SGI operation can be completed to a genuinely incoherent
operation by another SGI operation. Thus, a transformation among two
states can be implemented with some non-zero probability of success
if and only if there exists an SGI operation connecting them (up to
normalization). The following theorem generalizes Theorem \ref{thm:Schur}
to the stochastic scenario.
\begin{thm}
A quantum operation $\Lambda$ is SGI if and only if it can be written
as
\[
\Lambda[\rho]=A\odot\rho
\]
with a PSD matrix $A$ such that $0\leq A_{ii}\leq1$. \end{thm}
\begin{proof}
Proposition 4.17 and theorem 4.19 in \cite{Watrous2011} establish
the equivalence of Schur maps with a PSD matrix $A$ and maps with
diagonal Kraus operators independently of whether the maps are trace-preserving
or not. SGI maps correspond to the case of trace non-increasing maps.
Thus, it only remains to check that this condition is fulfilled if
and only if $0\leq A_{ii}\leq1$ $\forall i$. First of all it must
hold that $A_{ii}\geq0$ $\forall i$ in order for the matrix to be
PSD. Then, on the one hand, $A_{ii}\leq1$ $\forall i$ is clearly
sufficient for the Schur map to be trace non-increasing. On the other
hand, looking at the action of the map on the states $\{|i\rangle\langle i|\}$
the bound is also found to be necessary.
\end{proof}
The power of the above theorem lies in the fact that it gives a simple
characterization of all SGI operations. It provides the basis for
many results presented in this section.

\subsection{\label{coherenceranksec}Coherence rank and coherence set}

In entanglement theory \cite{HorodeckiRMP09}, local unitaries are
invertible local operations, and states related by local unitaries
have the same amount of entanglement. Thus, for most problems concerning
bipartite pure-state entanglement, it is sufficient to consider the
Schmidt coefficients of the corresponding states.

In direct analogy, we notice that diagonal unitaries are invertible
GI operations. Hence, for any measure of genuine coherence, states
related by diagonal unitaries are equally coherent. Thus, without
loss of generality, we can restrict our considerations to pure states
\begin{equation}
|\psi\rangle=\sum_{i}\psi_{i}|i\rangle
\end{equation}
such that $\psi_{i}\geq0$. Obviously, a pure state is incoherent
if and only if $\psi_{i}=1$ for some $i$.

In analogy to the Schmidt rank in entanglement theory~\cite{HorodeckiRMP09},
one can define the \emph{coherence rank} of a pure state $r(\ket{\psi})$
as the number of basis elements for which $\psi_{i}\neq0$~\cite{Killoran2015}.
The coherence rank, like its analogous in entanglement theory, provides
useful information about the coherence content of a state and constrains
the possible transformations among resource states. For instance,
the coherence rank cannot increase under incoherent operations \cite{Winter2015}.
One particularity of genuine coherence is, as we will see later, that
it is not only relevant the coherence rank but also for which basis
elements a state has zero components. Indeed, diagonal unitaries do
not allow to permute basis elements and by its very definition an
incoherent state cannot be transformed by GI operations into a different
incoherent state.
\begin{defn}
\label{def:coherence-set}The \emph{coherence set }$R(\psi)$ of $|\psi\rangle=\sum_{i=1}^{d}\psi_{i}|i\rangle$
denotes the subset of $\{1,2,\ldots,d\}$ for which $\psi_{i}\neq0$.
\end{defn}
Thus, unless otherwise stated, in the following when we start with
a state $|\psi\rangle=\sum_{i}\psi_{i}|i\rangle$ it should be assumed
that all sums go over the elements of $R(\psi)$. Finally, we will
always use the notation $\rho_{\psi}$ for the density matrix corresponding
to the pure state $|\psi\rangle$ (i.\ e.\ $\rho_{\psi}=|\psi\rangle\langle\psi|$).

With this, we have introduced all the machinery that we need to address
the main aim of this section. We proceed now by studying how GI operations
can be used to transform different quantum states into each other.

\subsection{State manipulation under GI operations}

\label{GImanipulationsec}

\subsubsection{Single-state transformations}

In the standard resource theory of coherence the state
\begin{equation}
|+_{d}\rangle=\frac{1}{\sqrt{d}}\sum_{i=1}^{d}|i\rangle\label{maxcoh}
\end{equation}
represents the golden unit: it can be transformed into any other state
on $\mathbb{C}^{d}$ via incoherent operations \cite{Baumgratz2014}.
It has been shown in \cite{Streltsov2015d} that this is no longer
the case for genuine coherence: via GI operations the state $|+_{2}\rangle$
cannot be deterministically converted into a different single-qubit
state $\ket{\psi}$ with $|\bk{\psi|0}|^{2}\neq1/2$. As we will show
in the following theorem, the situation is much more drastic.
\begin{thm}
\label{noconv} A pure state $\ket{\psi}$ can be deterministically
transformed into another pure state $\ket{\phi}$ via GI operations
if and only if $\ket{\phi}=U\ket{\psi}$ with a genuinely incoherent
unitary $U$.\end{thm}
\begin{proof}
From Theorem \ref{thm:Schur} it follows that for any GI operation
$\Lambda$ the states $\rho$ and $\Lambda[\rho]$ have the same diagonal
elements, i.e., $ $$\bk{i|\rho|i}=\bk{i|\Lambda[\rho]|i}$. For pure
states $\rho_{\psi}$ and $\rho_{\phi}=\Lambda[\rho_{\psi}]$ this
means that $|\bk{i|\psi}|^{2}=|\bk{i|\phi}|^{2}$, and thus the states
are the same up to a unitary $U$ which is diagonal in the incoherent
basis.
\end{proof}
This theorem shows that deterministic GI transformations among pure
states are trivial. We can only use invertible operations to transform
states within the classes of equally coherent states. This resembles
the case of multipartite entangled states where almost no pure state
can be transformed to any other state outside its respective equivalence
class \cite{DeVicente2013}.

We have seen that the framework of genuine coherence does not have
a golden unit, i.e., there is no unique state from which all other
states can be prepared via GI operations. However, it is still possible
that for every mixed state $\rho$ there exists some pure state $\ket{\psi}$
from which $\rho$ can potentially be created via GI operations. In
the following theorem we will show that this is indeed the case.
\begin{thm}
\label{convmix} For every mixed state $\rho$ there exists a pure
state $\rho_{\psi}$ and a GI operation $\Lambda$ such that
\begin{equation}
\rho=\Lambda[\rho_{\psi}].
\end{equation}
\end{thm}
\begin{proof}
Let $\rho=\sum_{ij}\rho_{ij}|i\rangle\langle j|$ be an arbitrary
mixed state. Since $\rho$ is PSD, we can assume that $\rho_{ii}\neq0$
for all $i$. We will now provide a pure state $\ket{\psi}$ and a
GI operation $\Lambda$ such that $\Lambda(\rho_{\psi})=\rho$. As
we will prove in the following, the desired state and GI operation
are given as
\begin{align}
|\psi\rangle & =\sum_{i}\sqrt{\rho_{ii}}|i\rangle,\\
\Lambda[X] & =\rho\odot\rho_{\tilde{\psi}}\odot X
\end{align}
with the matrix $\rho_{\tilde{\psi}}=\sum_{ij}(\sqrt{\rho_{ii}\rho_{jj}})^{-1}|i\rangle\langle j|$
\footnote{Note that $\rho_{\tilde{\psi}}$ is not normalized, i.e., $\mathrm{Tr}[\rho_{\tilde{\psi}}]\neq1$
in general}.

For proving this, we first note that the matrix $A=\rho\odot\rho_{\tilde{\psi}}$
is PSD, since it is the Schur product of two PSD matrices \cite{Horn1991}.
Notice moreover that $A_{ii}=1$ $\forall i$; hence, using Theorem
\ref{thm:Schur}, $\Lambda[X]=A\odot X$ is a GI operation. Since
$\rho_{\tilde{\psi}}$ is the Schur inverse of $\rho_{\psi}$ we finally
obtain that
\begin{equation}
\Lambda[\rho_{\psi}]=\rho\odot\rho_{\tilde{\psi}}\odot\rho_{\psi}=\rho,
\end{equation}
which is the desired result.
\end{proof}
This theorem shows that in the framework of genuine quantum coherence
the set of all pure states can be regarded as a resource: all mixed
states can be obtained from some pure states via GI operations. Thus,
although there is no maximally genuinely coherent state, there is
a maximal genuinely coherent set in the terminology of \cite{DeVicente2013}.
Moreover, noticing that transformations under GI operations require
that the diagonal entries of the density matrices are preserved, the
theorem further implies the following corollary, which characterizes
all conversions from pure states to mixed states.
\begin{cor}
\label{puretomixedGI} A pure state $|\psi\rangle$ can be deterministically
transformed by GI operations into the mixed state $\rho$ if and only
if $\bk{i|\rho_{\psi}|i}=\bk{i|\rho|i}$.
\end{cor}
At this point we also note that mixed states cannot be deterministically
transformed to a pure state. Indeed, let a mixed state have spectral
decomposition $\rho=\sum_{i}\lambda_{i}\rho_{\psi_{i}}$ with $0<\lambda_{i}<1$.
Then, if there existed a GI map $\Lambda$ such that $\Lambda(\rho)=\rho_{\phi}$
for some pure state $\rho_{\phi}$, we would need that $\Lambda(\rho_{\psi_{i}})=\rho_{\phi}$
$\forall i$, which is forbidden by Theorem \ref{noconv} (unless
$|\psi_{i}\rangle=U|\phi\rangle$ for all $i$ for some genuinely
incoherent unitary, which would imply that $\rho$ is pure).

The impossibility of deterministic GI conversions among pure states
calls for the analysis of probabilistic transformations. As explained
above this amounts to the use of SGI operations. In the following
theorem we evaluate the optimal probability for pure state conversion
via SGI operations. In this theorem we will also explicitly use Definition
\ref{def:coherence-set} for the coherence set $R$.
\begin{thm}
\label{thm:sgitrans} A probabilistic transformation by GI operations
from $|\psi\rangle$ to $|\phi\rangle$ is possible if and only if
$R(\phi)\subseteq R(\psi)$. The optimal probability of conversion
is
\begin{equation}
P(\rho_{\psi}\to\rho_{\phi})=\min_{i\in R(\phi)}\frac{\bk{i|\rho_{\psi}|i}}{\bk{i|\rho_{\phi}|i}}.
\end{equation}
\end{thm}
\begin{proof}
Without loss of generality we can write
\begin{eqnarray}
\ket{\psi}=\sum_{i\in R(\psi)}\psi_{i}\ket{i}, & \quad & \ket{\phi}=\sum_{j\in R(\phi)}\phi_{j}\ket{j}.
\end{eqnarray}
We will first show that the condition $R(\phi)\subseteq R(\psi)$
is necessary for a probabilistic transformation. Let $\Lambda(\cdot)=\sum_{i=1}^{n}K_{i}\cdot K_{i}^{\dag}$
be an SGI map such that $\Lambda(\rho_{\psi})\propto\rho_{\phi}$.
Then, it must be that $K_{i}|\psi\rangle\propto|\phi\rangle$ $\forall i$.
Hence, since the Kraus operators are diagonal, if for some index $k$
we have $\psi_{k}=0$ then $\phi_{k}=0$ must be true as well. This
proves that $R(\phi)\subseteq R(\psi)$ is a necessary condition for
probabilistic transformation.

We will now show that for $R(\phi)\subseteq R(\psi)$ there exists
a protocol implementing the transformation with the aforementioned
probability. For this, we additionally define the matrix
\begin{equation}
\rho_{\tilde{\psi}}=\sum_{i,j\in R(\psi)}(\psi_{i}\psi_{j})^{-1}\ket{i}\bra{j}
\end{equation}
and the number $c=\max_{i}(\phi_{i}/\psi_{i})^{2}$. Using again the
Schur product theorem, we have that the matrix $A=c^{-1}\rho_{\phi}\odot\rho_{\tilde{\text{\ensuremath{\psi}}}}$
is PSD with $A_{ii}\leq1$. Hence, there exists an SGI map $\Lambda$
such that $\Lambda(X)=A\odot X$. Moreover, $\Lambda(\rho_{\psi})=c^{-1}\rho_{\phi}$
and $\tr\Lambda(\rho_{\psi})=c^{-1}$. Thus, $|\psi\rangle$ can be
transformed to $|\phi\rangle$ with probability $c^{-1}$ (notice
that $1<c<\infty$).

In the final step, we show that there cannot exist a protocol with
larger probability of success. We do this by contradiction. Suppose
that there exists an SGI map $\Lambda$ with Schur representation
given by the PSD matrix $A$ such that $A\odot\rho_{\psi}=k\rho_{\phi}$
with $k>c^{-1}$. Let $i$ be the index for which $c=(\phi_{i}/\psi_{i})^{2}$.
Since $A_{ii}\psi_{i}^{2}=k\phi_{i}^{2}$, we would have that $A_{ii}=kc>1$,
which is in contradiction with the fact that $\Lambda$ is a trace
non-increasing map.
\end{proof}
Notice that the state $|+_{d}\rangle$ is not maximally genuinely
coherent even under the stochastic point of view. Indeed, let $|\chi\rangle$
be the state for which $\{\chi_{i}^{2}\}$ give rise to $\{1/2,(2(d-1))^{-1},\ldots,(2(d-1))^{-1}\}$.
Then, for $d>2$, $P(\rho_{\chi}\to\rho_{+})>P(\rho_{+}\to\rho_{\chi})$.
Given the impossibility to relate pure states by deterministic GI
operations, it would be tempting to order the set of pure coherent
states by $|\psi\rangle>|\phi\rangle$ if $P(\rho_{\psi}\to\rho_{\phi})>P(\rho_{\phi}\to\rho_{\psi})$.
Unfortunately, it turns out that such an order would be not well defined.
To see that, consider the state $|\psi\rangle$ with squared components
$\{1/4,5/8,1/8\}$. For $d=3$, we have that $P(\rho_{\chi}\to\rho_{+})>P(\rho_{+}\to\rho_{\chi})$
and $P(\rho_{+}\to\rho_{\psi})>P(\rho_{\psi}\to\rho_{+})$ but $P(\rho_{\psi}\to\rho_{\chi})>P(\rho_{\chi}\to\rho_{\psi})$.
On the other hand, by arguing as in \cite{Vidal2000}, we have that,
fixing a target state, the function $f_{\rho_{\phi}}(\rho_{\psi})=P(\rho_{\psi}\to\rho_{\phi})$
gives a computable genuine coherence monotone for all pure states.
Furthermore, $f_{\rho_{\phi}}(\rho_{\psi})=1$ iff $\ket{\phi}$ and
$\ket{\psi}$ are related via diagonal unitaries. Hence, by changing
the target state we obtain different monotones, each of them being
maximal for a different state in the maximal genuinely coherent set.

Finally, one may wonder whether mixed states can be transformed by
GI operations with some non-zero probability into a pure coherent
state. A simple example is given by $\rho=p|\psi\rangle\langle\psi|+(1-p)|\phi\rangle\langle\phi|$
where the two pure states $\ket{\psi}$ and $\ket{\phi}$ are respectively
supported on the orthogonal subspaces $W=\mathrm{span}\{|i\rangle\}_{i=1}^{n}$
and $W^{\perp}=\mathrm{span}\{|i\rangle\}_{i=n+1}^{d}$. If we denote
by $P_{X}$ the projector onto the subspace $X$, then the GI map
with Kraus operators $K_{1}=P_{W}$ and $K_{2}=P_{W^{\perp}}$ transforms
$\rho$ into $|\psi\rangle$ with probability $p$ and into $|\phi\rangle$
with probability $1-p$. The next theorem shows that this is essentially
the only possibility.
\begin{thm}
Let $W$ denote a subspace spanned by a subset of two elements of
the incoherent basis. Then, a non-pure coherent state $\rho$ can
be transformed by GI operations with non-zero probability to some
pure coherent state if and only if $P_{W}\rho P_{W}$ is a pure coherent
state for some choice of $W$. \end{thm}
\begin{proof}
The ``if'' part of the theorem is immediate since a map with a unique
Kraus operator given by $P_{W}$ is clearly an SGI operation. To prove
the ``only if'' part we will show that if $P_{W}\rho P_{W}$ is
not pure for any possible choice of $W$, then $\rho$ cannot be transformed
by GI operations with non-zero probability into a pure coherent state.
We will proceed by assuming the opposite and arriving at a contradiction.
Suppose that there exists an SGI map $\Lambda$ with Schur representation
given by the PSD matrix $A$ such that $A\odot\rho=k\rho_{\psi}$
with $0<k<1$ and $\rho_{\psi}=\sum_{i,j}\psi_{i}\psi_{j}|i\rangle\langle j|$.
By our premise, the projection of $\rho$ on every 2-dimensional subspace
spanned by two elements of the incoherent basis must be positive definite
(not PSD). This, together with the fact that $\rho$ is coherent implies
that there must exist $i,j$ such that $\rho_{ii},\rho_{jj},\rho_{ij}\neq0$
and $\rho_{ii}\rho_{jj}>|\rho_{ij}|^{2}$. The existence of the SGI
map imposes the following three equations
\begin{align*}
A_{ii}\rho_{ii} & =k\psi_{i}^{2},\\
A_{jj}\rho_{jj} & =k\psi_{j}^{2},\\
A_{ij}\rho_{ij} & =k\psi_{i}\psi_{j},
\end{align*}
which altogether yield
\[
|A_{ij}|^{2}=\frac{A_{ii}A_{jj}\rho_{ii}\rho_{jj}}{|\rho_{ij}|^{2}}>A_{ii}A_{jj}.
\]
However, this implies that $A$ cannot be PSD and, hence, a contradiction.
\end{proof}
Hence, most mixed states cannot be stochastically transformed to any
pure state. Thus, if, as discussed above, we regard pure states as
the most resourceful states over mixed states, it turns out that most
less resourceful states cannot be transformed, even with small probability,
to a resource state in the one-copy regime.

\subsubsection{Multiple-state and multiple-copy transformations}

So far we have just discussed possible transformations acting on a
single copy of a state. However, in quantum information theory it
is standard to find that multiple-state transformations broaden the
possibilities for resource manipulation. An important example of this
are activation phenomena. This means that the transformation $\rho\otimes\sigma\to\rho'\otimes junk$
(or, more generally, $\rho\otimes\sigma\to\tau$ with $\tr_{junk}\tau=\rho'$)
is possible even though it is impossible to implement the conversion
$\rho\to\rho'$. In this case the state $\sigma$ is called an activator.

Another example, and probably the most paradigmatic one, is distillation.
In these protocols one aims at transforming many copies of a less
useful state into less copies of a maximally useful state in the asymptotic
limit of infinitely many available copies. For instance, in entanglement
theory this target state that acts as a golden standard to measure
the usefulness of the resource is the maximally entangled two-qubit
state $|\Phi^{+}\rangle=(|00\rangle+|11\rangle)/\sqrt{2}$.

In general, we say that a state $\sigma$ can be distilled from the
state $\rho$ at rate $0<R\leq1$ if $\rho^{\otimes n}\to\tau$ and
$\tr_{junk}\tau$ is $\varepsilon$-close to $\sigma^{\otimes nR}$
and $\varepsilon\to0$ as $n\to\infty$. As a measure of closeness
we will use the trace norm; that is, it must hold that $\Vert\tr_{junk}\tau-\sigma^{\otimes nR}\Vert\leq\varepsilon$
where $||M||=\tr\sqrt{M^{\dagger}M}$. The optimal rate at which distillation
is possible, i.\ e.\ the supremum of $R$ over all protocols fulfilling
the aforementioned conditions, is a very relevant figure of merit
known as distillable resource and plays a key role for the quantification
of usefulness in resource theories. The reversed protocol, which is
known as dilution, is also an interesting object of study. In this
case one seeks for the optimal rate at which less copies of maximally
useful state can be converted into more copies of a less useful state.
This leads to another figure of merit: the resource-cost. In more
detail, the cost of $\rho$ is the infimum of the rate $R$ over all
protocols with $0<R\leq1$ such that $\sigma^{\otimes nR}\otimes junk$
(where $\sigma$ is a golden unit maximally resourceful state) transforms
$\varepsilon$-close to $\rho^{\otimes n}$ and $\varepsilon\to0$
as $n\to\infty$. The distillable entanglement and entanglement-cost
have been widely studied in entanglement theory \cite{HorodeckiRMP09}
and allow to establish the phenomenon of irreversibility. More recently,
the distillable coherence and coherence-cost have been characterized
and irreversibility has also been identified in this setting \cite{Winter2015}.

In order to discuss multiple-state and multiple-copy manipulation
under GI operations, it should be made clear what the set of allowed
maps is in this setting. If we are allowed to act jointly on $n$
different states each of them acting on the Hilbert space $H\simeq\mathbb{C}^{d}$,
we define the incoherent basis in the total Hilbert space $H^{\otimes n}$
as $\{|i_{1}i_{2}\cdots i_{n}\rangle\}$ ($i_{j}=1,\ldots,d$ $\forall j$),
where $\{|i_{j}\rangle\}$ is the incoherent basis in each Hilbert
space \cite{Bromley2015,Streltsov2015b}. This can be further justified
by the no superactivation postulate (cf.\ Ref.\ \cite{Chitambar2016}).
Thus, joint GI operations should preserve incoherent states in this
basis and they will be characterized by having Kraus operators diagonal
in the joint incoherent basis. By the same reasons as in Section \ref{sec:Schur},
these GI maps will also admit a Schur representation in the joint
incoherent basis.

We are now in the position to state our results on multiple-state
and multiple-copy manipulation under joint GI operations. It turns
out that these protocols are out of reach: activation and any non-trivial
form of distillation and dilution are impossible. This claim is a
consequence of the following lemma.
\begin{lem}
\label{noactivation} For every two states $\rho,\sigma\in\mathcal{L}(\mathbb{C}^{d})$
and every GI map $\Lambda$ acting on $\mathcal{L}(\mathbb{C}^{d}\otimes\mathbb{C}^{d})$
such that $\Lambda(\rho\otimes\sigma)=\tau$, there exists another
GI map $\tilde{\Lambda}$ acting on $\mathcal{L}(\mathbb{C}^{d})$
such that $\tilde{\Lambda}(\rho)=\tau_{1}:=\tr_{2}(\tau)$. \end{lem}
\begin{proof}
By assumption together with Theorem \ref{thm:Schur}, there exists
a PSD matrix A with diagonal entries equal to 1,
\begin{equation}
A=\sum_{ijkl}A_{ik,jl}|ik\rangle\langle jl|\quad(A_{ik,ik}=1\,\forall i,k),
\end{equation}
which induces the GI operation $\Lambda$ such that
\begin{equation}
\tau=\Lambda(\rho\otimes\sigma)=A\odot(\rho\otimes\sigma)=\sum_{ijkl}\rho_{ij}\sigma_{kl}A_{ik,jl}|ik\rangle\langle jl|.
\end{equation}
The state $\tau_{1}$ is obtained by taking the partial trace over
the second subsystem:
\begin{equation}
\tau_{1}=\mathrm{tr}_{2}(\tau)=\sum_{ij}\left(\sum_{k}\sigma_{kk}A_{ik,jk}\right)\rho_{ij}|i\rangle\langle j|=\tilde{A}\odot\rho
\end{equation}
with the matrix $\tilde{A}$ with entries $\tilde{A}_{ij}=\sum_{k}\sigma_{kk}A_{ik,jk}$.
The proof is complete if we can show that the operator $\tilde{A}$
is PSD and that $\tilde{A}_{ii}=1$ holds for all $i$, since in this
case by Theorem \ref{thm:Schur} there must exist a GI operation $\tilde{\Lambda}$
such that $\tilde{\Lambda}(\rho)=\tilde{A}\odot\rho$. It is straightforward
to verify that $\tilde{A}_{ii}=1$ holds true for all $i$. To see
that $\tilde{A}$ is PSD, notice that $\tilde{A}=\tr_{2}(X\odot A)$,
where the operator $X$ can be written as
\begin{equation}
X=\sum_{ijkl}\sigma_{kl}|ik\rangle\langle jl|=\proj{v}\otimes\sigma
\end{equation}
with the (unnormalized) vector $\ket{v}=\sum_{i}\ket{i}$. Thus, $X$
is PSD (and so is $A$ by assumption). Hence, by the Schur product
theorem, $\tilde{A}$ is the partial trace of a PSD matrix and it
must me PSD too.
\end{proof}
The above lemma shows that there cannot be any activation phenomena
in the resource theory of genuine coherence and it will be very useful
in our study of coherence distillation and dilution with GI operations.
In order to analyze the possibility of distillation one first needs
to discuss what the target state is going to be. However, this is
not at all clear under GI operations since, as we pointed out in the
previous subsection, there is no unique state which would allow to
create all other states via GI operations. In the following, we will
show that in general it is not possible to distill the state $\sigma$
from $\rho$ via GI operations if $\sigma$ has more coherence than
$\rho$. Here, we measure the coherence by the relative entropy of
coherence \cite{Baumgratz2014}
\begin{equation}
C_{r}(\rho)=\min_{\sigma\in\mathcal{I}}S(\rho||\sigma)
\end{equation}
with the quantum relative entropy $S(\rho||\sigma)=\mathrm{tr}(\rho\log_{2}\rho)-\mathrm{tr}(\rho\log_{2}\sigma)$
and $\mathcal{I}$ denotes the set of all incoherent states. The relative
entropy of coherence is known to be equal to the distillable coherence
\cite{Winter2015}, and $C_{r}$ is also a faithful genuine coherence
monotone \cite{Streltsov2015d}. We are now in position to prove the
following theorem.
\begin{thm}
Given two states $\rho$ and $\sigma$ with
\begin{equation}
C_{r}(\rho)<C_{r}(\sigma),
\end{equation}
it is not possible to distill $\sigma$ from $\rho$ at any rate $R>0$
via GI operations.\end{thm}
\begin{proof}
We will prove the statement by contradiction, assuming that distillation
is possible for some state $\rho$ with $C_{r}(\rho)<C_{r}(\sigma)$.
In particular, this would imply that for $n$ large enough it is possible
to approximate one copy of the state $\sigma$. To be more precise,
for any $\varepsilon>0$ there exists an integer $n$ and a GI operation
$\Lambda$ such that
\begin{equation}
\left\Vert \tr_{n-1}\left(\Lambda\left[\rho^{\otimes n}\right]\right)-\sigma\right\Vert \leq\varepsilon,
\end{equation}
where the partial trace is taken over some subset of $n-1$ copies.

In the next step we use Lemma \ref{noactivation} to note that the
map $\tr_{n-1}\left(\Lambda\left[\rho^{\otimes n}\right]\right)$
can always be written as a GI operation $\tilde{\Lambda}$ acting
on just one copy of $\rho$:
\begin{equation}
\tilde{\Lambda}\left[\rho\right]=\tr_{n-1}\left(\Lambda\left[\rho^{\otimes n}\right]\right).
\end{equation}
Combining the aforementioned arguments, we conclude that for any $\varepsilon>0$
there exists a GI operation $\tilde{\Lambda}$ such that
\begin{equation}
\left\Vert \tilde{\Lambda}\left[\rho\right]-\sigma\right\Vert \leq\varepsilon.
\end{equation}

In the final step we will use the asymptotic continuity of the relative
entropy of coherence (see Lemma~12 in \cite{Winter2015}). It implies
that
\begin{equation}
\left|C_{r}\left(\tilde{\Lambda}\left[\rho\right]\right)-C_{r}\left(\sigma\right)\right|\leq\varepsilon\log_{2}d+2h\left(\frac{\varepsilon}{2}\right)
\end{equation}
with the binary entropy $h(x)=-x\log_{2}x-(1-x)\log_{2}(1-x)$ and
$d$ the fixed dimension of the Hilbert space. Using the fact that
the bound on the right-hand side of this inequality is continuous
for $\varepsilon\in(0,1)$ and that it goes to zero as $\varepsilon\to0$,
we can say that for any $\delta>0$ there exists some GI operation
$\tilde{\Lambda}$ such that
\begin{equation}
\left|C_{r}\left(\tilde{\Lambda}\left[\rho\right]\right)-C_{r}\left(\sigma\right)\right|\leq\delta.\label{eq:continuity}
\end{equation}
On the other hand, the assumption $C_{r}(\rho)<C_{r}(\sigma)$ implies
that there exists some $\delta>0$ such that
\begin{equation}
C_{r}\left(\sigma\right)-C_{r}\left(\rho\right)\geq\delta.
\end{equation}
Recalling that the relative entropy of coherence is a genuine coherence
monotone, i.e., $C_{r}(\tilde{\Lambda}[\rho])\leq C_{r}(\rho)$, we
arrive at the following result:
\begin{equation}
C_{r}\left(\sigma\right)-C_{r}\left(\tilde{\Lambda}[\rho]\right)\geq\delta
\end{equation}
for some $\delta>0$ and any GI operation $\tilde{\Lambda}$. This
is a contradiction to Eq.~(\ref{eq:continuity}), and the proof of
the theorem is complete.
\end{proof}
From the above theorem it follows that it is not possible to distill
the state $\ket{+_{2}}$ from any non-equivalent single-qubit state
$\rho$, since $C_{r}(\rho)<1=C_{r}(\ket{+_{2}})$.

In the final part of this section we address the impossibility of
dilution. Interestingly, it turns out that diluting less copies of
a state into more copies of another is generically impossible independently
of which state is picked as a golden unit.
\begin{thm}
Given any two coherent states $\rho$ and $\sigma$ of the same dimensionality
it is not possible to dilute $\sigma$ to $\rho$ at any rate $R<1$
via GI operations.\end{thm}
\begin{proof}
If dilution at rate $R<1$ (i.\ e.\ leaving aside one-copy deterministic
transformations when possible) was possible, this would require that
$\forall\varepsilon>0$ there existed integers $m<n$ such that
\begin{equation}
||\Lambda(\sigma^{\otimes m}\otimes junk)-\rho^{\otimes n}||=||\tau-\rho^{\otimes n}||\leq\varepsilon.
\end{equation}
Notice that the presence of some junk incoherent part is indispensable
as GI operations cannot increase the dimensionality. Since the trace
distance cannot increase by quantum operations, by tracing out the
$m$ particles in the system the above equation requires in particular
that
\begin{equation}
||\tr_{system}\tau-\rho^{\otimes(n-m)}||\leq\varepsilon.
\end{equation}
Now, Lemma \ref{noactivation} implies that $\tr_{system}\tau=\tilde{\Lambda}(junk)$
for some GI map $\tilde{\Lambda}$. Hence, it must hold that
\begin{equation}
||\tilde{\Lambda}(junk)-\rho^{\otimes(n-m)}||\leq\varepsilon.
\end{equation}
However, the junk part is incoherent and, therefore, so must be $\tilde{\Lambda}(junk)$.
Any coherent state is bounded away from the set of incoherent states
and, thus, the above inequality cannot hold $\forall\varepsilon>0$
for any coherent state $\rho$.
\end{proof}
Thus, dilution with rate $R<1$ from a more useful qudit state into
a less useful qudit state is impossible by GI operations independently
of the measure of coherence used.

\section{\label{SecFI}Fully incoherent operations}

\subsection{General concept}

As discussed in the introduction, one of the main reasons to consider
GI operations was to rule out any form of hidden coherence in the
free operations of a resource theory of coherence. However, we have
seen in the previous section that state manipulation under GI operations
might be too limited. This leads to think whether one could consider
a larger set of allowed operations still having the property that
every Kraus representation is incoherent but that could allow for
a richer structure for state manipulation.

We recall that the incoherent operations considered in the resource
theory of coherence introduced by Baumgratz \emph{et al.} in \cite{Baumgratz2014}
are given by Kraus operators $\{K_{i}\}$ such that for every incoherent
state $\rho$, $K_{i}\rho K_{i}^{\dag}$ is (up to normalization)
an incoherent state as well $\forall i$. It will be relevant in the
following to notice that such $\{K_{i}\}$ are characterized by having
at most one non-zero entry in every column \cite{Du2015b}. The fact
that the Kraus operators of a different Kraus representation of some
incoherent operation might not be incoherent can easily be seen \cite{Streltsov2015d}
by using the following well-known theorem (see e.\ g.\ \cite{Watrous2011,Holevo2012}).
\begin{thm}
\label{kraus} Two sets of Kraus operators
$\{K_{j}\}$ and $\{L_{i}\}$ correspond to Kraus representations
of the same map if and only if there exists a partial isometry matrix $V$
such that
\begin{equation}
L_{i}=\sum_{j}V_{ij}K_{j}.\label{krauseq}
\end{equation}

\end{thm}
On the contrary, using Theorem \ref{kraus} it is straightforward
to check that GI maps are incoherent (in fact, even diagonal) in every
Kraus representation. As discussed above, it comes as a natural question
whether GI maps constitute the most general class of operations having
this property. Interestingly, as we will show in the following, the
answer is no: there exist operations which are not GI but still incoherent
in every Kraus representation. A simple example for such an operation
is the erasing map, which puts every input state onto the state $|0\rangle\langle0|$.
This operation is incoherent in every Kraus representation because
it must hold that $K_{i}\rho K_{i}^{\dag}\propto|0\rangle\langle0|$
for every Kraus operator. On the other hand, the operation is clearly
not GI because any incoherent state that is not $|0\rangle\langle0|$
does not remain invariant under this map. We will call this class
of operations which are incoherent in every Kraus representation \textbf{\emph{f}}\emph{ully
}\textbf{\emph{i}}\emph{ncoherent (FI).}

One might wonder then what are the properties of the class of FI maps
and which differences it has with the class of GI maps. In particular,
we want to compare these sets of operations in the task of state transformation.
For this, we will first provide a full characterization of FI operations
in the following theorem.
\begin{thm}
\label{thm:FI} A quantum operation is FI if and only if all Kraus
operators are incoherent and have the same form.
\end{thm}
\noindent Before we prove the theorem some remarks are in place. The
requirement that all Kraus operators have the same form means that
their nonzero entries are all at the same position (i.\ e.\ whenever
there is a non-zero entry in a given column, it must occur at the
same row for every Kraus operator). As an example, according to the
theorem any single-qubit quantum operation defined by the Kraus operators
\begin{equation}
K_{1}=\left(\begin{array}{cc}
a & b\\
0 & 0
\end{array}\right),\,\,\,K_{2}=\left(\begin{array}{cc}
c & d\\
0 & 0
\end{array}\right)\label{eq:FIexample}
\end{equation}
is fully incoherent, since both Kraus operators are incoherent and
have the same form. The completeness condition $\sum_{i}K_{i}^{\dagger}K_{i}=\openone$
puts constraints on the complex parameters: $|a|^{2}+|c|^{2}=|b|^{2}+|d|^{2}=1$
and $a^{*}b+c^{*}d=0$. Note that -- according to the theorem -- this
map is FI, but it is not GI since the Kraus operators are not diagonal.
Indeed, this is exactly the erasing map which maps every state onto
$\ket{0}\bra{0}$. We will now provide the proof of Theorem \ref{thm:FI}.
\begin{proof}
That maps with this property are FI is immediate. If the Kraus operators
in one representation are all incoherent and have a particular given
form, then, by Theorem \ref{kraus}, so does every Kraus representation
since Eq.\ (\ref{krauseq}) preserves this structure. Hence, every
Kraus representation is incoherent.

Thus, to complete the proof we only need to see that any map which
has one (incoherent) Kraus representation in which not all operators
are of the same form cannot be FI, i.\ e.\ it must then admit another
Kraus representation which is not incoherent. For such, we can take
without loss of generality that
\begin{equation}
col_{1}(K_{1})=\left(\begin{array}{c}
\ast\\
0\\
0\\
\vdots\\
0
\end{array}\right),\quad col_{1}(K_{2})=\left(\begin{array}{c}
0\\
\ast\\
0\\
\vdots\\
0
\end{array}\right),
\end{equation}
where $col_{i}(A)$ denotes the $i$-th column of the matrix $A$
and $\ast$ an arbitrary non-zero number. Moreover, we define two
unitary matrices $V$ and $U$:
\begin{equation}
V=\left(\begin{array}{cc}
U & 0\\
0 & \openone
\end{array}\right),\quad U=\frac{1}{\sqrt{2}}\left(\begin{array}{rr}
1 & 1\\
1 & -1
\end{array}\right).
\end{equation}
Since $V$ is unitary the set of Kraus operators $\{L_{i}\}$ constructed
using Eq.\ (\ref{krauseq}) is another Kraus representation of the
map given by $\{K_{i}\}$. However, we find that
\begin{equation}
col_{1}(L_{1})=\left(\begin{array}{c}
\ast\\
\ast\\
0\\
\vdots\\
0
\end{array}\right),
\end{equation}
and the representation given by $\{L_{i}\}$ is therefore not incoherent.
\end{proof}
From Theorem \ref{thm:FI} it is clear that GI maps are contained
in the class of FI maps since every Kraus operator in every Kraus
representation is diagonal, hence fulfilling the condition of Theorem
\ref{thm:FI}. As also noted above Theorem~\ref{thm:FI}, there exist
operations which are FI but not GI, and the erasing map $\Lambda[\rho]=\proj{0}$
is one example. Another instance of FI maps, taking for example maps
on $\mathcal{L}(\mathbb{C}^{3})$, are those for which the Kraus operators
are given by
\begin{equation}
K_{i}=\left(\begin{array}{ccc}
a_{i} & 0 & 0\\
0 & 0 & 0\\
0 & b_{i} & c_{i}
\end{array}\right).
\end{equation}

A property fulfilled by FI maps is that these operations allow to
prepare any pure incoherent state from an arbitrary input state. For
that one would use the corresponding erasure map with the property
$\Lambda[\rho]=\proj{i}$ for every state $\rho$. It can be easily
seen that any such map is always FI for any choice of incoherent state
$\proj{i}$. This appears to be a reasonable property from the point
of view of a resource theory since it is usually considered that free
states can be prepared with the free operations. Notice, however,
that this is not possible with GI operations as all incoherent states
are fixed points of these maps (e.\ g.\ starting from the initial
state $\ket{0}$, GI operations do not allow to create any other incoherent
state $\ket{i}$). Nevertheless, as discussed in the introduction,
the physical setting might impose this constraint. For instance, this
property can be interpreted as energy-preservation: if the states
$\ket{i}$ are the eigenstates of some nondegenerate Hamiltonian,
GI operations are incoherent operations which additionally preserve
the energy of the system.

Another feature of FI maps is that any incoherent unitary transformation
can be implemented. Thus, elements of the incoherent basis can be
permuted and the coherence set is no longer meaningful. Hence, the
coherence rank $r$ takes the relevant role instead, similarly to
state manipulation under incoherent operations.

On the other hand, a striking feature of FI maps is that they constitute
a non-convex set. More precisely, given two FI operations $\Lambda_{1}$
and $\Lambda_{2}$, their convex combination
\begin{equation}
\Lambda[\rho]=p\Lambda_{1}[\rho]+(1-p)\Lambda_{2}[\rho]\label{eq:nonconvex}
\end{equation}
is not always fully incoherent. This can be demonstrated with the
following single-qubit operations:
\begin{align}
\Lambda_{1}[\rho] & =\sigma_{x}\rho\sigma_{x},\\
\Lambda_{2}[\rho] & =\sigma_{z}\rho\sigma_{z}
\end{align}
with the Pauli matrices
\begin{equation}
\sigma_{x}=\left(\begin{array}{cc}
0 & 1\\
1 & 0
\end{array}\right),\,\,\,\sigma_{z}=\left(\begin{array}{cc}
1 & 0\\
0 & -1
\end{array}\right).
\end{equation}
While both operations $\Lambda_{1}$ and $\Lambda_{2}$ are fully
incoherent, their convex combination $\Lambda$ in Eq.~(\ref{eq:nonconvex})
is not fully incoherent for $0<p<1$. This can be seen by noting that
the Kraus operators of $\Lambda$ are given by $K_{1}=\sqrt{p}\sigma_{x}$
and $K_{2}=\sqrt{1-p}\sigma_{z}$. Since these Kraus operators do
not have the same form for $0<p<1$, by Theorem \ref{thm:FI} the
operation cannot be fully incoherent.

This non-convexity is, thus, a consequence of the fact that FI maps
are characterized by a property of the set of implemented Kraus operators
and not of each individual operator as it is the case for incoherent
or GI operations. In practice, this means that if one considers state
manipulation under FI maps, it turns out that two particular operations
might be allowed but not to implement each of them with a given probability.
In particular, this implies that, although pure incoherent states
can be prepared with the free operations as pointed out above, this
result does not need to extend to mixed incoherent states despite
the fact that they constitute free states as well. This is because
it is not allowed to mix different FI fixed-output maps. Actually,
it can be checked that it is not the case that every state can be
transformed to any mixed incoherent state by some FI map. The interested
reader is referred to Appendix \ref{sec:appendixB} for a particular
example. This construction relies on an observation on the structure
of FI maps used in Theorem \ref{thm:FIconv} below.

\subsection{\label{sub:Permutations}Permutations as basic FI operations}

As mentioned above, any genuinely incoherent operation is also fully
incoherent. Here, we will consider another important class of FI operations,
which we call \emph{permutations}. A permutation $P_{ij}$ with $i\neq j$
is a unitary which interchanges the states $\ket{i}$ and $\ket{j}$,
and preserves all states $\ket{k}$ for $k\neq i,j$:
\begin{align}
 & P_{ij}\ket{i}=\ket{j},\\
 & P_{ij}\ket{j}=\ket{i},\\
 & P_{ij}\ket{k}=\ket{k}\,\,\mathrm{for\,all}\,\,k\neq i,j.
\end{align}
The above definition involves the permutation of only two states $\ket{i}$
and $\ket{j}$. In the following, we will also consider more general
permutations with more than two elements. We will denote an arbitrary
general permutation by $P$. Any such permutation can be decomposed
as a product of permutations of only two states. Notice that any such
$P$ corresponds to an FI unitary transformation.

\subsection{State manipulation under FI operations}

\subsubsection{Pure state deterministic transformations}

In this section we study the potential for deterministic state manipulation
if the allowed set of operations is given by FI maps. Interestingly,
we will see in the following that -- contrary to the case of GI maps
-- transformations among pure states are possible. However, these
are rather limited as shown in the following theorem.
\begin{thm}
\label{thm:FIconv} A deterministic FI transformation from $|\psi\rangle$
to $|\phi\rangle$ is possible only if $r(\phi)\leq r(\psi)$. Moreover,
for $r(\phi)=r(\psi)$ a transformation is possible if and only if
$\ket{\phi}=U\ket{\psi}$ with a fully incoherent unitary $U$.\end{thm}
\begin{proof}
The first part of the theorem is true due to the more general result
stating that the coherence rank cannot increase under incoherent operations
as we already mentioned in Sec.\ \ref{coherenceranksec}. Since FI
operations are a subclass of general incoherent operations, the coherence
rank also cannot increase under FI operations.

We will now prove the second part of the theorem, assuming that $\ket{\psi}$
and $\ket{\phi}$ are two states with the same coherence rank, and
that
\begin{equation}
\rho_{\phi}=\Lambda[\rho_{\psi}]
\end{equation}
with some FI operation $\Lambda$. Since both states have the same
coherence rank, the FI operation $\Lambda$ must contain at least
one Kraus operator $K$ which has one nonzero element in each row
and column. Moreover, all remaining Kraus operators must have the
same shape as $K$, i.e., their nonzero elements must be at the same
positions as the nonzero elements of $K$. It is now crucial to note
that any such map must be a composition of a GI operation $\Lambda_{\mathrm{gi}}$
and a permutation $P$:
\begin{equation}
\Lambda[\rho]=P\Lambda_{\mathrm{gi}}[\rho]P^{\dagger}.
\end{equation}
Now, since we assume that $\Lambda[\rho_{\psi}]$ is pure, it must
be that $\Lambda_{\mathrm{gi}}[\rho_{\psi}]$ is also pure, since
the permutation $P$ is a unitary operation. By Theorem \ref{noconv},
the fact that $\Lambda_{\mathrm{gi}}[\rho_{\psi}]$ is pure implies
that $\Lambda_{\mathrm{gi}}[\rho_{\psi}]=U_{\mathrm{gi}}\rho_{\psi}U_{\mathrm{gi}}^{\dagger}$
with some diagonal unitary $U_{\mathrm{gi}}$. Combining these results
we arrive at the following expression:
\begin{equation}
\rho_{\phi}=\Lambda[\rho_{\psi}]=P\Lambda_{\mathrm{gi}}[\rho_{\psi}]P^{\dagger}=PU_{\mathrm{gi}}\rho_{\psi}U_{\mathrm{gi}}^{\dagger}P^{\dagger}.
\end{equation}
The proof of the theorem is complete by noting that $PU_{\mathrm{gi}}$
is a fully incoherent unitary.
\end{proof}
If we apply now the above theorem to study single-qubit FI operations
from an arbitrary single-qubit state $|\psi\rangle$, we see that
the only possible output for transformations to a pure state is either
an incoherent state or any state which is equivalent to $|\psi\rangle$
under FI unitaries. This shows that FI operations are strictly less
powerful than general incoherent operations in their ability to convert
pure quantum states into each other. In particular, general incoherent
operations can convert the state $\ket{+_{2}}$ into any other single-qubit
state \cite{Baumgratz2014}. This is not possible via FI operations,
as follows from the above discussion by taking $\ket{\psi}=\ket{+_{2}}$.

Having seen that FI operations are less powerful than general incoherent
operations, we will now show that FI operations are still more powerful
than GI operations. For this, we can use a fixed-output map to see
that by FI operations it is possible to transform the state $\ket{+_{2}}$
into the state $\ket{0}$. By Theorem \ref{noconv} this process is
not possible with GI operations. However, this is a transformation
to a non-resource state. One might wonder whether FI operations allow
for nontrivial transformations to a pure coherent state. In the following,
we provide such an example. Consider a FI map with Kraus operators
given by
\begin{equation}
K_{1}=\left(\begin{array}{ccc}
a_{1} & 0 & c_{1}\\
0 & b_{1} & 0\\
0 & 0 & 0
\end{array}\right),\quad K_{2}=\left(\begin{array}{ccc}
a_{2} & 0 & c_{2}\\
0 & b_{2} & 0\\
0 & 0 & 0
\end{array}\right).
\end{equation}
The normalization condition $\sum_{i}K_{i}^{\dagger}K_{i}=\openone$
imposes that
\begin{align}
a_{1}c_{1}^{*}+a_{2}c_{2}^{*} & =0,\nonumber \\
|x_{1}|^{2}+|x_{2}|^{2} & =1,\quad x=a,b,c.\label{eq:normalization}
\end{align}
In order for the (normalized) state
\begin{equation}
|\psi\rangle=\left(\begin{array}{c}
\psi_{1}\\
\psi_{2}\\
\psi_{3}
\end{array}\right)\in\mathbb{R}^{3}
\end{equation}
to be convertible with this map into another pure state, the only
requirement is that $K_{1}|\psi\rangle=kK_{2}|\psi\rangle$ for some
constant $k\in\mathbb{C}$. As can be verified by inspection, this
condition is equivalent to
\begin{equation}
\psi_{3}=\frac{a_{2}b_{1}-a_{1}b_{2}}{b_{2}c_{1}-b_{1}c_{2}}\psi_{1}.\label{eq:condition}
\end{equation}

Now, a proper choice of the free parameters can be made so that the
conditions (\ref{eq:normalization}) and (\ref{eq:condition}) are
met. For instance, we can take $b_{1}=b_{2}=1/\sqrt{2}$, $a_{1}=c_{2}=\sqrt{3}/2$
and $a_{2}=-c_{1}=1/2$, which leads to
\begin{equation}
\psi_{3}=\frac{(\sqrt{3}-1)^{2}}{2}\psi_{1}.\label{eq:condition-2}
\end{equation}
With a real choice of $\psi_{1}$ small enough the above equation
can be fulfilled with a properly normalized state. Indeed, if we further
choose $\psi_{1}=1/2$, condition (\ref{eq:condition-2}) implies
that $\psi_{3}=1-\sqrt{3}/2$. The remaining parameter $\psi_{2}$
is restricted by the normalization of the state $\ket{\psi}$, and
can be chosen as $\psi_{2}=(1-\psi_{1}^{2}-\psi_{3}^{2})^{1/2}=(\sqrt{3}-1)^{1/2}$.
The pure final state $\Lambda[\rho_{\psi}]=\rho_{\phi}$ is then given
by
\begin{equation}
\ket{\phi}=\frac{\sqrt{6}-\sqrt{2}}{2}\ket{0}+(\sqrt{3}-1)^{1/2}\ket{1}.
\end{equation}
This example shows that FI operations allow for nontrivial transformations
between pure states which are not possible with genuinely incoherent
operations. Still, as we will see in the next theorem with the particular
example of the state $|+_{3}\rangle$, it seems that FI transformations
among pure coherent states are nevertheless very constrained.
\begin{thm}
\label{thm:QutritConvertibility}Via deterministic FI operations,
the state $\ket{+_{3}}$ can be transformed into one of the following
pure states:
\begin{equation}
\Lambda[\ket{+_{3}}\bra{+_{3}}]=\begin{cases}
\proj{0}\\
\proj{\psi}\\
\proj{+_{3}}
\end{cases}\label{eq:Qutrits}
\end{equation}
with $\ket{\psi}=\sqrt{2/3}|0\rangle+\sqrt{1/3}|1\rangle$. Moreover,
up to FI unitaries, this is the full set of pure states which can
be obtained from $\ket{+_{3}}$ via deterministic FI operations. \end{thm}
\begin{proof}
We will prove the statement by studying states with a fixed coherence
rank that can be obtained from the state $\ket{+_{3}}$ via FI operations.
We start with states of coherence rank 1, i.e., incoherent states
$\ket{i}$. The state $\ket{0}$ can be obtained from the state $\ket{+_{3}}$
via the erasing operation which maps every state onto $\ket{0}$.
As mentioned earlier, this operation is fully incoherent. Any incoherent
state $\ket{i}$ can be obtained from $\ket{+_{3}}$ by first performing
the erasing operation, and then applying a fully incoherent unitary
which transforms $\ket{0}$ to $\ket{i}$.

Before we study states of coherence rank 2, we first consider the
case of coherence rank 3 for simplicity. Since the initial state $\ket{+_{3}}$
has coherence rank 3, we can apply Theorem \ref{thm:FIconv}, stating
that FI transformation between states with the same coherence rank
must be necessarily unitary, i.e., the final state must be $\ket{+_{3}}$
or FI unitary equivalent of it.

We now come to the most difficult part of the proof, where we will
show that for final states of coherence rank 2, the only possible
pure state is given by
\begin{equation}
\ket{\psi}=\sqrt{2/3}|0\rangle+\sqrt{1/3}|1\rangle\label{eq:psi}
\end{equation}
and FI unitary equivalents of it. In particular, we will show that
no other pure state of coherence rank 2 can be obtained from $\ket{+_{3}}$
via FI operations.

According to Theorem \ref{thm:FI}, the only FI maps that are able
to produce states of coherence rank 2 must have Kraus operators of
the form
\begin{equation}
K_{i}=P\left(\begin{array}{ccc}
a_{i} & 0 & c_{i}\\
0 & b_{i} & 0\\
0 & 0 & 0
\end{array}\right)P',
\end{equation}
where $P$ and $P'$ are arbitrary (but fixed $\forall i$) permutations.
Since permutations are FI unitaries and the state $\ket{+_{3}}$ is
invariant under any permutation it is thus sufficient to study Kraus
operators of the form
\begin{equation}
K_{i}=\left(\begin{array}{ccc}
a_{i} & 0 & c_{i}\\
0 & b_{i} & 0\\
0 & 0 & 0
\end{array}\right).
\end{equation}
The completeness condition $\sum_{i}K_{i}^{\dagger}K_{i}=\openone$
imposes that
\begin{align}
\sum_{i}a_{i}c_{i}^{\ast} & =0,\nonumber \\
\sum_{i}|z_{i}|^{2} & =1,\quad z=a,b,c.
\end{align}
Each outcome of such FI maps leads to the following unnormalized state
\begin{equation}
K_{i}|+_{3}\rangle\propto\left(\begin{array}{c}
a_{i}+c_{i}\\
b_{i}\\
0
\end{array}\right).
\end{equation}
Assume now that a deterministic transformation to the state $|\phi\rangle$
is possible. This means that $K_{i}|+_{3}\rangle\propto|\phi\rangle$
$\forall i$. In particular, this implies that $|\phi\rangle\propto\sum_{i}K_{i}|+_{3}\rangle$
and, thus
\begin{equation}
|\phi\rangle\propto\left(\begin{array}{c}
\sum_{i}x_{i}\\
\sum_{i}b_{i}\\
0
\end{array}\right)\quad(x_{i}=a_{i}+c_{i}),
\end{equation}
with the above normalization conditions imposing that $\sum_{i}|x_{i}|^{2}=2$
and $\sum_{i}|b_{i}|^{2}=1$. Moreover, the fact that $K_{i}|+_{3}\rangle\propto K_{j}|+_{3}\rangle$
$\forall i\neq j$ demands that $x_{i}=k_{i}x$ and $b_{i}=k_{i}b$
for some complex numbers $\{k_{i}\}$, $x$ and $b$. Thus,
\begin{equation}
|\phi\rangle\propto\sum_{i}k_{i}\left(\begin{array}{c}
x\\
b\\
0
\end{array}\right)\propto\left(\begin{array}{c}
x\\
b\\
0
\end{array}\right).
\end{equation}
Furthermore, the normalization conditions then lead to $|x|^{2}=2|b|^{2}$.
This means that it must hold that
\begin{equation}
|\phi\rangle\propto\left(\begin{array}{c}
\sqrt{2}e^{i\alpha}\\
1\\
0
\end{array}\right).
\end{equation}
After normalization this state is FI unitary equivalent to the state
$\ket{\psi}$ given in Eq.~(\ref{eq:psi}). These arguments prove
that apart from $\ket{\psi}$ and FI unitary equivalents no other
pure state with coherence rank 2 can be obtained from $\ket{+_{3}}$
via FI operations.

To see that the transformation to $\ket{\psi}$ is indeed possible,
we choose an FI operation given by the following two Kraus operators:
\begin{align}
K_{1} & =\left(\begin{array}{ccc}
i/\sqrt{2} & 0 & 1/\sqrt{2}\\
0 & 1/\sqrt{2} & 0\\
0 & 0 & 0
\end{array}\right),\nonumber \\
K_{2} & =\left(\begin{array}{ccc}
1/\sqrt{2} & 0 & i/\sqrt{2}\\
0 & 1/\sqrt{2} & 0\\
0 & 0 & 0
\end{array}\right).
\end{align}
This map acting on $|+_{3}\rangle$ leads to the state $\sqrt{2/3}e^{i\pi/4}|0\rangle+\sqrt{1/3}|1\rangle$,
which can be then transformed to the aforementioned state with an
FI unitary.
\end{proof}
Theorem \ref{thm:QutritConvertibility} provides severe constraints
on the pure state conversion via FI operations. In particular, it
shows that via FI operations the state $\ket{+_{3}}$ can only be
converted into three different pure states, and their FI unitary equivalents.
Being FI operations a particular class of general incoherent operations,
it seems natural to compare their power. In fact, the latter class
of transformations shows a much richer structure. The question whether
two pure states can be converted into each other via incoherent operations
has been recently addressed in \cite{Du2015b,Winter2015}. The answer
is closely related to the corresponding problem in entanglement theory
which was solved by Nielsen in \cite{Nielsen1999} relying on the
theory of majorization. For two density matrices $\rho$ and $\sigma$
with corresponding eigenvalues $\{p_{i}\}_{i=1}^{d}$ and $\{q_{j}\}_{j=1}^{d}$,
the majorization relation $\rho\succ\sigma$ means that the inequality
\begin{equation}
\sum_{i=1}^{t}p_{i}\geq\sum_{j=1}^{t}q_{j}
\end{equation}
holds true for all $t\leq d$.

It has been shown in \cite{Winter2015} that given two pure states
$\ket{\psi}$ and $\ket{\phi}$ with the same coherence rank, there
exists an incoherent operation transforming $\ket{\psi}$ into $\ket{\phi}$
if and only if $\Delta[\rho_{\phi}]\succ\Delta[\rho_{\psi}]$. Here,
$\Delta[\rho]=\sum_{i}\bk{i|\rho|i}\proj{i}$ denotes full dephasing
in the incoherent basis. The earlier reference \cite{Du2015b} states
this same result without the assumption that both states have the
same coherence rank. However, as pointed out recently by Winter and
Yang \cite{Winter2015} and by Chitambar and Gour (see page 9 in \cite{Chitambar2016}),
the part of the proof showing the necessity of the majorization condition
for a transformation to be possible has serious flaws. Notwithstanding,
its sufficiency is clearly established. With this we can compare the
power of incoherent and FI transformations. Theorem \ref{thm:FIconv}
already shows a limitation of the latter. While for any pair of incoherent-unitary
inequivalent states of the same coherence rank, the transformation
$|\psi\rangle\to|\phi\rangle$ is possible by incoherent operations
whenever $\Delta[\rho_{\phi}]\succ\Delta[\rho_{\psi}]$ is fulfilled,
this is not possible with FI operations. Moreover, $\Delta[\rho_{\phi}]\succ\Delta[\rho_{+_{3}}]$
for any qutrit-state $|\phi\rangle$ and, therefore, $|+_{3}\rangle$
can be transformed to any qutrit-state by incoherent operations. However,
we have seen in Theorem \ref{thm:QutritConvertibility} that the only
coherence-rank-two state obtainable from this state by deterministic
FI manipulation is $\sqrt{2/3}|0\rangle+\sqrt{1/3}|1\rangle$ (and
its FI unitary equivalents). Thus, this shows that even under the
constraint that the coherence rank decreases the majorization condition
is far from being a sufficient condition for FI transformations.

\subsubsection{Pure state stochastic transformations}

One may ask for the potential of FI operations for stochastic manipulation
of states. Being GI operations a subset of FI operations, the corresponding
optimal protocols of Theorem \ref{thm:sgitrans} provide lower bounds
on the maximal probability of conversion under SFI operations. In
general, these protocols are suboptimal as, for example, from the
previous section we know that $|+_{3}\rangle\to\sqrt{2/3}|0\rangle+\sqrt{1/3}|1\rangle$
can be realized in this case with probability one. Notice, however,
that these bounds can be strengthened since the FI setting allows
to apply a permutation to the input state before implementing the
protocol. With this observation and the insight of Theorem \ref{thm:sgitrans}
we can characterize the optimal probability for SFI transformations
among states with the same coherence rank and establish non-trivial
lower bounds for the case of coherence-rank-decreasing operations
(obviously, a transformation to a state with a higher coherence rank
cannot be accomplished with nonzero probability).
\begin{thm}
The optimal probability of transforming a state into another by FI
operations fulfills
\begin{equation}
P(\rho_{\psi}\to\rho_{\phi})\geq\max_{P}\min_{i}\frac{\bk{i|P\rho_{\psi}P^{\dag}|i}}{\bk{i|\rho_{\phi}|i}},
\end{equation}
where $P$ is any permutation (and the minimization goes over all
$i$ such that $\bk{i|\rho_{\phi}|i}\neq0$ if $r(\psi)>r(\phi)$).
Moreover, the inequality becomes an equality if $r(\psi)=r(\phi)$. \end{thm}
\begin{proof}
As discussed above, the bound is obvious from Theorem \ref{thm:sgitrans}
as we can use the protocol of the GI setting optimized over all FI
unitary equivalents of the input state. To see that this gives the
actual optimal probability when the input and output state have the
same coherence rank, notice that in this case at least one Kraus operator
of the corresponding SFI map must be full-rank. Since all Kraus operators
must be of the same form, this means that they all must take the form
$K_{i}=D_{i}P$ with some permutation $P$ and diagonal matrices $D_{i}$.
Thus, the most general SFI maps to achieve such transformation are
given by $\Lambda_{\mathrm{gi}}(P\rho_{\psi}P^{\dag})$ where $\Lambda_{\mathrm{gi}}$
is any SGI map and $P$ is any permutation.
\end{proof}

\subsubsection{Many-copy transformations}

The most striking limitation of state manipulation under GI operations
is the generic impossibility of distillation and dilution. These are
a consequence of the obstruction to any form of activation proven
in Lemma \ref{noactivation}. It would be very interesting to see
whether, although deterministic one-copy transformations are rather
limited as well for FI operations, distillation and dilution protocols
are possible in this case. Interestingly, we show in the following
that Lemma \ref{noactivation} is no longer true for FI operations.
This leaves open the possibility that many-copy transformations in
this case might have a rich structure.

To prove our claim, it suffices to consider a counterexample. Take
the obvious extension to $\mathbb{C}^{4}$ of the protocol implementing
$|+_{3}\rangle\to\sqrt{2/3}|0\rangle+\sqrt{1/3}|1\rangle$ to see
that $|+_{4}\rangle\to(\sqrt{2}|0\rangle+|1\rangle+|3\rangle)/2$
is possible by FI. For this, a map with Kraus operators given by
\begin{align}
K_{1} & =\left(\begin{array}{cccc}
i/\sqrt{2} & 0 & 1/\sqrt{2} & 0\\
0 & 1/\sqrt{2} & 0 & 0\\
0 & 0 & 0 & 0\\
0 & 0 & 0 & 1/\sqrt{2}
\end{array}\right),\nonumber \\
K_{2} & =\left(\begin{array}{cccc}
1/\sqrt{2} & 0 & i/\sqrt{2} & 0\\
0 & 1/\sqrt{2} & 0 & 0\\
0 & 0 & 0 & 0\\
0 & 0 & 0 & 1/\sqrt{2}
\end{array}\right),
\end{align}
together with an FI unitary transformation does the job. This shows
that $|+_{2}\rangle^{\otimes2}\to(\sqrt{2}|00\rangle+|01\rangle+|11\rangle)/2$
can be done with two copies. If we now trace out the second qubit
we see that this activates the transformation $|+_{2}\rangle\to\rho$
where
\begin{equation}
\rho=\left(\begin{array}{cc}
3/4 & 1/4\\
1/4 & 1/4
\end{array}\right).
\end{equation}
This transformation is clearly impossible with one copy because of
similar arguments used before. Since $\rho$ is full-rank, the transformation
$|+_{2}\rangle\to\rho$ could only be implemented by an FI operation
with full-rank Kraus operators. This means that we can only use compositions
of permutations and GI maps but Corollary \ref{puretomixedGI} tells
us that this would require the diagonal entries of $|+_{2}\rangle\langle+_{2}|$
and $\rho$ to be equal (up to permutations), which is not the case.

\section{Relation to other concepts of quantum coherence}

In this section we will discuss the relation of genuinely incoherent
operations and fully incoherent operations to other concepts of coherence.
We start with a list of alternative incoherent operations recently
proposed in the literature:
\begin{itemize}
\item Maximally incoherent operations (MIO) \cite{Aberg2006}: operations
which preserve the set of incoherent states, i.e., $\Lambda(\rho)\in\mathcal{I}$
for all $\rho\in\mathcal{I}$.
\item Dephasing-Covariant incoherent operations (DIO) \cite{Chitambar2016,Marvian2016}:
operations which commute with dephasing, i.e., $\Delta[\Lambda(\rho)]=\Lambda(\Delta[\rho])$.
\item Translationally invariant operations (TIO) \cite{Gour2009,Marvian2014a,Marvian2015}:
operations which commute with the unitary translation $e^{-itH}$
for some nondegenerate Hermitian operator $H$ diagonal in the incoherent
basis, i.e., $\Lambda[e^{-itH}\rho e^{itH}]=e^{-itH}\Lambda[\rho]e^{itH}$
\footnote{Some references (cf.\ Refs.\ \cite{Chitambar2016,Marvian2016})
allow for degenerate operators $H$. This is a legitimate choice which
is relevant in certain scenarios. However, this induces a different
set of free states and leads to a different resource theory in which
coherence is measured relative to the eigenspaces of $H$. This is
why we do not consider here this possibility.}.
\item Strictly incoherent operations (SIO) \cite{Winter2015,Yadin2015a}:
operations with a Kraus decomposition $\{K_{i}\}$ such that each
Kraus operator commutes with dephasing, i.e., $\Delta[K_{i}\rho K_{i}^{\dagger}]=K_{i}\Delta[\rho]K_{i}^{\dagger}$.
\item Physical incoherent operations (PIO) \cite{Chitambar2016}: maps with
Kraus operators of the form $K_{j}=\sum_{x}e^{i\theta_{x}}|\pi_{j}(x)\rangle\bra{x}P_{j}$
and their convex combinations. Here, $\pi_{j}$ are permutations and
$P_{j}$ is an orthogonal and complete set of incoherent projectors.
\end{itemize}
The following inclusions have been proven in the aforementioned references
(see e.g. \cite{Chitambar2016})
\begin{align}
\mathrm{PIO} & \subset\mathrm{SIO}\subset\mathrm{DIO}\subset\mathrm{MIO,}\\
\mathrm{PIO} & \subset\mathrm{SIO}\subset\mathrm{IO}\subset\mathrm{MIO,}.%\\
%\mathrm{IO} & \nsubset\mathrm{DIO},\,\,\,\,\,\,\,\,\mathrm{DIO}\nsubset\mathrm{IO.}
\end{align}

In the following we will show that FIO and SIO are not subsets of
each other, i.e.,
\begin{equation}
\mathrm{FIO}\nsubset\mathrm{SIO},\,\,\,\,\,\,\,\,\,\,\,\,\,\,\,\,\,\mathrm{SIO}\nsubset\mathrm{FIO.}
\end{equation}
For proving $\mathrm{FIO}\nsubset\mathrm{SIO}$, note that SI operations
with Kraus operators $\{K_{i}\}$ are characterized by the operators
$K_{i}^{\dagger}$ being also incoherent $\forall i$ \cite{Winter2015}.
This amounts to the fact that the Kraus operators not only have at
most one non-zero entry per column but also at most one non-zero entry
per row. FI operations with non-diagonal rank-deficient Kraus operators
clearly fail to fulfill this requirement. Notice that this also implies
that $\mathrm{FIO}\nsubset\mathrm{PIO}$. Finally, it is easy to see
that $\mathrm{SIO}\nsubset\mathrm{FIO}$, since the Kraus operators
of $\mathrm{SIO}$ need not have the same form.

We proceed to show that FI operations are a subset of DI operations,
i.e.,
\begin{equation}
\mathrm{FIO}\subset\mathrm{DIO}.
\end{equation}
Lemma 17 of \cite{Chitambar2016} states that $\Lambda$ is DIO if
and only if $\Lambda(|x\rangle\langle x|)$ is incoherent and $\Delta(\Lambda(|x\rangle\langle x'|)=0$
$\forall x\neq x'$. The first condition is clearly fulfilled by FI
maps. To see that the second condition also holds, notice that the
$(j,j)$ entry of $\Lambda(|x\rangle\langle x'|)$ is given by $\sum_{i}K_{i}(j,x)K_{i}^{\ast}(j,x')$
if $\Lambda$ has Kraus operators $\{K_{i}\}$. For FI maps, either
$K_{i}(j,x)K_{i}^{\ast}(j,x')=0$ $\forall i$ (and the second condition
is then trivially fulfilled for the $(j,j)$ entry) or $K_{i}(m,x),K_{i}(m,x')=0$
$\forall i$ and $\forall m\neq j$. In this last case the $(x,x')$
entry of the normalization condition $\sum_{i}K_{i}^{\dag}K_{i}=\openone$
reads then precisely $\sum_{i}K_{i}^{\ast}(j,x)K_{i}(j,x')=0$, as
we wanted to prove.

Moreover, it is easy to see the following relations:
\begin{align}
\mathrm{FIO} & \subset\mathrm{IO},\\
\mathrm{GIO} & \subset\mathrm{SIO},\\
\mathrm{PIO} & \nsubset\mathrm{GIO}.
\end{align}
While $\mathrm{FIO}\subset\mathrm{IO}$ is obvious, $\mathrm{GIO}\subset\mathrm{SIO}$
follows from the fact that genuinely incoherent Kraus operators are
all diagonal. Finally, $\mathrm{PIO}\nsubset\mathrm{GIO}$ holds true
because the Kraus operators of PIO need not be diagonal, and that
PI operations allow for distillation \cite{Chitambar2016}. For qubit
and qutrit maps it further holds that $\mathrm{GIO}\subset\mathrm{PIO}$,
because in this situation GI operations are mixed unitaries, see Appendix
\ref{sec:MathematicalStructure}. However, this inclusion breaks down in higher dimensions and, in general, we have that $\mathrm{GIO}\nsubset\mathrm{PIO}$ as well. To see this, consider the GI map with Kraus operators given by
\begin{equation}
K_1=diag(1,0,\cos\theta,\cos\theta),\quad K_2=diag(0,1,sin\theta,i\sin\theta)
\end{equation}
for some value of $\theta\in(0,\pi/2)$. It has been shown in \cite{Chitambar2016} that PIO correspond to convex combinations of maps with Kraus operators of the form $K_j=U_jP_j$ $\forall j$, where the $U_j$ are incoherent unitaries and the $P_j$ form an orthogonal and complete set of incoherent projections. If a PI map is also GI without loss of generality we can take the $U_j$ to be diagonal. Thus, according to Theorem \ref{kraus}, if the above map was also in PIO, it is necessary that there exists a linear combination of the Kraus operators $L=\alpha K_1+\beta K_2$ ($\alpha,\beta\neq0$) such that for every $i\in\{1,2,3,4\}$ we either have $L_{ii}=0$ or $|L_{ii}|=\sqrt{p}$ for some fixed $p\in(0,1]$. Since
\begin{equation}
L=diag(\alpha,\beta,\alpha\cos\theta+\beta\sin\theta,\alpha\cos\theta+i\beta\sin\theta),
\end{equation}
the above condition can only hold for $i=1,2$ if $|\alpha|=|\beta|=\sqrt{p}$. This furthermore imposes that both $|L_{33}|=\sqrt{p}$ and $|L_{44}|=\sqrt{p}$ since we cannot have then that these entries vanish. However, the first condition then leads to $\textrm{Re}(\alpha^*\beta)=0$ and the second to $\textrm{Im}(\alpha^*\beta)=0$. Hence, we would need that $\alpha=\beta=0$, which cannot be. Thus, the given GI map is not in PIO.

Last, it holds that FIO and TIO are not subsets of each other. To see that
$\mathrm{FIO}\nsubset\mathrm{TIO}$, it suffices to note that permutations
are not TIO since the only unitary operations in TIO are diagonal
\cite{Marvian2016}. In order to prove that $\mathrm{TIO}\nsubset\mathrm{FIO}$,
one can consider the fully depolarizing qubit map, which maps every
qubit state to $\openone/2$. This map is clearly in TIO and has Kraus
operators $K_{i}=\sigma_{i}/2$ with $i=0,1,2,3$ where $\sigma_{0}$
is the identity and the others are the standard Pauli matrices. Since
the Kraus operators do not have the same form the fully depolarizing
map cannot be an FI operation. It was further proven in \cite{Streltsov2015d}
that $\mathrm{GIO}\subset\mathrm{TIO}$, where the inclusion is strict.

\section{Conclusions}

In this paper we provide a detailed study of the framework of genuine
coherence introduced in ref. \cite{Streltsov2015d}. We found a general
characterization of genuinely incoherent operations via Schur maps,
and used this result to prove strong limitations of this framework.
In particular, genuinely incoherent operations do not have a golden
unit, and also do not allow for asymptotic state transformation or
dilution. On the single-copy level, genuinely incoherent operations
only allow for transformations to states with the same diagonal elements.
Nevertheless, more general transformations are possible via stochastic
GI operations. One of our main results it the optimal probability
for this procedure in Theorem~\ref{thm:sgitrans}.

We also introduced and studied the class of fully incoherent operations:
these are all operations which are incoherent in any Kraus decomposition.
We provide a complete characterization of this set of operations,
and show that it is strictly larger than the set of genuinely incoherent
operations. On the other hand, we show that state transformations
are also very limited in this framework: there is no ``maximally
coherent state'' which would allow for transformations to all other
states, and deterministic pure state transformations are only possible
between very restricted families of states. It remains open, however,
whether distillation and dilution procedures are generally possible
in this framework. Since pure incoherent states can be transformed
into each other via fully incoherent operations, this concept characterizes
speakable coherence.

The physical setting in other resource theories such as entanglement
clearly defines the set of free operations. However, in the current
efforts to develop a resource theory of coherence, although it is
clear that free states should correspond to incoherent states, the
notion of free operations is not completely clear from the physical
context. Actually, it has been recently discussed \cite{Chitambar2016,Marvian2016}
that the incoherent operations considered in the standard resource
theory of coherence of Baumgratz \emph{et al}. \cite{Baumgratz2014},
although mathematically consistent, have not been provided with a
clear physical interpretation. In this article we have considered
two possible sets of free operations for a resource theory of coherence
that stem from a clear physical motivation. GI operations leave all
incoherent states invariant and, thus, account for physical situations
in which the free states are not interchangeable like, for instance,
in the presence of energy constraints (this situation is relevant
in several applications and it is referred to as unspeakable coherence,
cf.\ \cite{Marvian2016}). On the other hand, FI operations are the
most general class of maps that do not allow for hidden coherence,
i.\ e.\ they are incoherent irrespectively of the Kraus representation.
Our results show that these sets of free operations are very restricted
for deterministic state manipulation. This could lead to the conclusion
that any reasonably powerful resource theory of quantum coherence
must contain free operations that allow to interchange incoherent
states and that can potentially create coherence in some experimental
realization. This is in particular true if the theory should have
a golden unit, i.e., a single state from which all other states can
be created for free. This result leads to the question which additional
minimal requirements a resource theory of coherence must have, apart
from the aforementioned property. In particular, is it enough to add
a finite number of free operations which can potentially create coherence?
We leave this question open for future research.
\begin{acknowledgments}
The authors are very grateful to Eric Chitambar for suggesting the example that shows that $\mathrm{GIO}\nsubset\mathrm{PIO}$. AS acknowledges funding by the Alexander von Humboldt-Foundation.
The research of JIdV was funded by the Spanish MINECO through grants
MTM2014-54692 and MTM2014-54240-P and by the Comunidad de Madrid through
grant QUITEMAD+CM S2013/ICE-2801.
\end{acknowledgments}

\appendix
%dummy comment inserted by tex2lyx to ensure that this paragraph is not empty%dummy comment inserted by tex2lyx to ensure that this paragraph is not empty%dummy comment inserted by tex2lyx to ensure that this paragraph is not empty%dummy comment inserted by tex2lyx to ensure that this paragraph is not empty

\section{\label{sec:MathematicalStructure}Mathematical structure of the set
of GI maps}

A genuinely incoherent operation $\Lambda_{\mathrm{gi}}$ will be
called \emph{extremal in the set of GI operations} if it cannot be
written as
\begin{equation}
\Lambda_{\mathrm{gi}}[\rho]=p\Lambda_{\mathrm{gi}}'[\rho]+(1-p)\Lambda_{\mathrm{gi}}''[\rho],
\end{equation}
with GI operations $\Lambda'_{\mathrm{gi}}\neq\Lambda''_{\mathrm{gi}}$
and probability $0<p<1$. It is straightforward to see that the set
of genuinely incoherent operations is compact and convex, which ensures
that any GI operation can be written as a convex combination of extremal
GI operations \cite{Rockafellar1970}.

A unital operation $\Lambda_{\mathrm{u}}$ is an operation which preserves
the maximally mixed state: $\Lambda_{\mathrm{u}}(\openone/d)=\openone/d$.
It is easy to see that any GI operation is unital, but there exist
unital operations which are not GI. We will call a GI operation \emph{extremal
in the set of unital operations} if it cannot be written as
\begin{equation}
\Lambda_{\mathrm{gi}}[\rho]=p\Lambda_{\mathrm{u}}'[\rho]+(1-p)\Lambda_{\mathrm{u}}''[\rho],
\end{equation}
with some unital operations $\Lambda_{\mathrm{u}}'\neq\Lambda_{\mathrm{u}}''$
and probability $0<p<1$.

Clearly, a GI operation which is extremal in the set of unital operations
must also be extremal in the set of GI operations. As we will see
in the following theorem, the inverse of this statement is true as
well.
\begin{lem}
\label{thm:extremal}A GI map is extremal in the set of GI maps if
and only if it is extremal in the set of unital maps.\end{lem}
\begin{proof}
As mentioned above the lemma, the implication from right to left is
clear. To see the other direction, we show that every map which is
not extremal in the set of unital maps is also not extremal in the
set of GI maps. Assume that $\Lambda$ is a GI map which is not extremal
in the set of unital maps, i.\ e.\ there exist unital maps $\Phi$
and $\Phi'$ such that $\Lambda=p\Phi+(1-p)\Phi'$ for some $p\in(0,1)$.
This means that if $\{K_{i}\}$ ($\{K'_{j}\}$) is a Kraus representation
of $\Phi$ ($\Phi'$), then $\{\sqrt{p}K_{i},\sqrt{1-p}K'_{j}\}$
is a Kraus representation of $\Lambda$. Since every Kraus representation
of a GI map is diagonal, the operators $\{K_{i}\}$ and $\{K'_{j}\}$
have to be then all diagonal. Thus, $\Phi$ and $\Phi'$ must be GI
maps as well and $\Lambda$ is therefore not extremal in the set of
GI maps.
\end{proof}
This lemma reveals a close connection between GI operations and unital
operations, and will be used to prove several important statements
on the structure of these maps in the following.

It has been noted in \cite{Streltsov2015d} that mixed-unitary maps,
i.\ e.
\begin{equation}
\Lambda(\rho)=\sum_{i}p_{i}U_{i}\rho U_{i}^{\dagger}\label{eq:MU}
\end{equation}
with $\{p_{i}\}$ probabilities and $\{U_{i}\}$ unitary matrices,
are GI maps when the matrices $U_{i}$ are all diagonal. It has been
moreover shown in \cite{Streltsov2015d} that the converse is true
for $d=2$, i.\ e.\ every GI operation on $\mathcal{L}(\mathbb{C}^{2})$
is a mixed-unitary map. However, the question was left open for $d>2$.
We will solve this question in the following. The case $d=3$ turns
out to be equivalent to exercise 4.1 in \cite{Watrous2011}. For the
sake of completeness we provide a full proof in the next theorem.
\begin{thm}
\label{thm:qutrit}Every GI operation on $\mathcal{L}(\mathbb{C}^{3})$
is mixed-unitary, i.e., can be written as in Eq.~(\ref{eq:MU}). \end{thm}
\begin{proof}
We recall from above that every GI operation can be written as a convex
combination of extremal GI operations. Clearly, diagonal unitaries
are extremal GI operations. Thus, the proof is complete if we show
that diagonal unitaries are the only extremal GI operations for qutrits.

For this, let $\Lambda(\cdot)=\sum_{i}K_{i}\cdot K_{i}^{\dag}$ be
a Kraus decomposition of an arbitrary GI map, where without loss of
generality the diagonal $\{K_{i}\}$ are taken to be linearly independent
(i. e. a so-called minimal representation \cite{Holevo2012}). Since
a map is then unitary if and only if $|i|=1$, we have to show that
$\Lambda$ is not extremal whenever $|i|\geq2$. It is shown in Theorem
4.21 in \cite{Watrous2011} that a unital map on $\mathcal{L}(\mathbb{C}^{d})$
given by a set of linearly independent Kraus operators $\{K_{i}\}$
is extremal in the set of unital maps if and only if the collection
of $|i|^{2}$ $d^{2}\times d^{2}$ matrices
\begin{equation}
\mathcal{K}_{ij}=\left(\begin{array}{cc}
K_{i}^{\dag}K_{j} & 0\\
0 & K_{j}K_{i}^{\dag}
\end{array}\right)\label{eq:K}
\end{equation}
is linearly independent. Notice that in the case of GI operations
it holds that $[K_{i}^{\dag},K_{j}]=0$ and, therefore, all matrices
$\mathcal{K}_{ij}$ are of the form $D_{ij}\oplus D_{ij}$ for diagonal
matrices $D_{ij}=K_{i}^{\dagger}K_{j}$. In our case $d=3$, the matrices
$\mathcal{K}_{ij}$ span then at most a 3-dimensional subspace. However,
if $|i|\geq2$ we have at least four matrices $\mathcal{K}_{ij}$,
and it is hence impossible that all of them are linearly independent.
Thus, $\Lambda$ is not extremal in the set of unital maps whenever
$|i|\geq2$. Lemma \ref{thm:extremal} ensures that a GI map which
is not extremal in the set of unital maps is also not extremal in
the subset of GI maps.
\end{proof}
The above theorem shows that all GI operations for a single qutrit
can be written as a mixture of diagonal unitaries, thus extending
the same statement for qubits found in \cite{Streltsov2015d}. It
is now natural to ask if this result also extends to higher dimensions.
This is not the case for dimension 4 or larger.
\begin{thm}
For every $d\geq4$ there exists a GI map which is not mixed-unitary.\end{thm}
\begin{proof}
We will first show this for $d=4$, and extend it to all dimensions
$d\geq4$ below. We will prove the statement by presenting a GI map
which is extremal but not unitary, and thus cannot be written as a
mixture of diagonal unitaries. According to Lemma \ref{thm:extremal},
it suffices to check extremality in the set of unital maps.

For this, we define two linearly independent Kraus operators $K_{1}=\mathrm{diag}(a_{1},a_{2},a_{3},a_{4})$
and $K_{2}=\mathrm{diag}(b_{1},b_{2},b_{3},b_{4})$ with $a_{k}=1/k$
and $b_{k}=i^{k}\sqrt{1-a_{k}^{2}}$ for $k=1,2,3,4$. It is easy
to check that these two Kraus operators define a genuinely incoherent
quantum operation. We will now show that this operation is extremal,
and thus cannot be written as a mixture of unitaries. For this, we
will use similar arguments as in the proof of Theorem \ref{thm:qutrit}.
In particular, we now have four $16\times16$ matrices $\mathcal{K}_{ij}$
given by the Kraus operators $K_{1}$ and $K_{2}$ via Eq.~(\ref{eq:K}).
The proof is complete by proving that these four matrices are linearly
independent.

For proving this we note that these matrices are directly related
to the following four vectors:
\begin{equation}
t_{i}=|a_{i}|^{2},\quad u_{i}=|b_{i}|^{2},\quad v_{i}=\overline{a_{i}}b_{i},\quad w_{i}=a_{i}\overline{b_{i}}.
\end{equation}
In particular, $\mathcal{K}_{11}=\mathrm{diag}(\boldsymbol{t})\oplus\mathrm{diag}(\boldsymbol{t})$,
$\mathcal{K}_{22}=\mathrm{diag}(\boldsymbol{u})\oplus\mathrm{diag}(\boldsymbol{u})$,
$\mathcal{K}_{12}=\mathrm{diag}(\boldsymbol{v})\oplus\mathrm{diag}(\boldsymbol{v})$,
and $\mathcal{K}_{21}=\mathrm{diag}(\boldsymbol{w})\oplus\mathrm{diag}(\boldsymbol{w})$.
For completing the proof, it is thus enough to show that these four
vectors are linearly independent. This can be done by verifying that
the determinant of the matrix $(\boldsymbol{t},\boldsymbol{u},\boldsymbol{v},\boldsymbol{w})$
is nonzero.

We will now show how the above arguments can also be used for dimensions
$d\geq4$. In this case we define two Kraus operators $K_{1}=\mathrm{diag}(a_{1},\ldots,a_{d})$
and $K_{2}=\mathrm{diag}(b_{1},\ldots,b_{d})$, where $a_{i}$ and
$b_{i}$ are defined in the same way as above for $i\leq4$. For $i>4$
we define $a_{i}=1$ and $b_{i}=0$ \footnote{We chose $a_{i}=1$ and $b_{i}=0$ for simplicity, but any other choice
of parameters $a_{i}$ and $b_{i}$ satisfying $|a_{i}|^{2}+|b_{i}|^{2}=1$
for $i>4$ will also lead to the desired result.}. It can be verified by inspection that the previous arguments also
apply in this situation, thus leading to a genuinely incoherent operation
which is extremal but not unitary. This completes the proof for all
dimensions $d\geq4$ .
\end{proof}
These considerations complete our investigation on the structure of
GI operations.

\section{\label{sec:appendixB}FI maps do not allow to prepare arbitrary mixed
incoherent states from arbitrary input states}

As argued in Sec.\ \ref{SecFI}, pure incoherent states can always
be obtained from any state by some FI operation. In this section we
also mentioned that for mixed incoherent states this is no longer
the case even though they are free states as well. To see a particular
example, take any $d$-dimensional state $\rho$ such that it does
not hold that $\rho_{ii}=1/d$ $\forall i$. It turns out that any
such state cannot be mapped by FI operations to the $d$-dimensional
maximally mixed state $\openone_{d}/d$, which is incoherent. The
reason is the following. Clearly, $\rho$ cannot be mapped by GI operations
to $\openone_{d}/d$ (Corollary \ref{puretomixedGI}). As in the proof
of Theorem \ref{thm:FIconv}, this also excludes any FI map with Kraus
operators whose form is any row permutation of a diagonal matrix.
This is because these maps can be written as $P\Lambda_{\mathrm{gi}}[\rho]P^{\dagger}$
for an arbitrary permutation $P$ and GI map $\Lambda_{gi}$. However,
since the output of the map, $\openone_{d}/d$, is left invariant
under permutations, we would need that $\Lambda_{\mathrm{gi}}[\rho]=\openone_{d}/d$,
which is impossible as argued above. Thus, the remaining possible
FI maps must be such that their Kraus operators are rank-deficient.
However, since Kraus operators of FI maps have all the same form,
it cannot be that the outputs of such maps are supported in the full
$d$-dimensional space. This shows that the transformation $\rho\to\openone_{d}/d$
with $\rho$ as defined above cannot be implemented by FI operations.

 \bibliographystyle{apsrev4-1}
\bibliography{literature}

\end{document}